\documentclass[preprintnumbers,floatfix,superscriptaddress,nofootinbib,prx,twocolumn,showpacs]{revtex4-1}
\usepackage{amsmath,amsthm,amssymb}
\usepackage{graphicx}
\usepackage[english]{babel} 
\usepackage[T1]{fontenc} 
\usepackage[utf8]{inputenc} 
\usepackage{graphicx}

\usepackage{fancyhdr}
\usepackage{xcolor}
\usepackage{braket}
\usepackage{amsmath}
\usepackage{amssymb}
\usepackage{amsthm}
\usepackage{enumitem}
\usepackage{multirow}
\usepackage{lipsum}
\usepackage{changepage}
\usepackage{physics}
\usepackage{framed}
\usepackage{wasysym}
\usepackage{hyperref}

\usepackage{floatrow}
\newfloatcommand{capbtabbox}{table}[][\FBwidth]

\usepackage{dsfont}
\usepackage{pifont}
\newcommand{\tn}{\textnormal}

\definecolor{cinnabar}{rgb}{0.89, 0.26, 0.2}

\newcommand{\vqss}{\texttt{VQSS}}
\newcommand{\vhss}{\texttt{VHSS}}
\newcommand{\vcss}{\texttt{VCSS}}

\newcommand{\good}{2-\small{GOOD}\normalsize}
\newcommand{\hybrid}{\{$p,t,n$\}-\texttt{VHSS}}
\newcommand{\ramp}{\{$p,t,t',n$\}-ramp \texttt{VHSS}}

\theoremstyle{definition}

\newtheorem{definition}{Definition}
\newtheorem{assumption}{Assumption}
\newtheorem{proposition}{Proposition}
\theoremstyle{plain}
\newtheorem{lemma}{Lemma}
\newtheorem{theorem}{Theorem}

\usepackage{array}
\newcolumntype{C}[1]{>{\centering\let\newline\\\arraybackslash\hspace{0pt}}m{#1}}

\begin{document}

\title{Verifiable Hybrid Secret Sharing With Few Qubits}
\date{\today}

\author{Victoria Lipinska}\email{v.lipinska@tudelft.nl}
\affiliation{QuTech, Delft University of Technology, Lorentzweg 1, 2628 CJ Delft, The Netherlands}
\affiliation{Kavli Institute of Nanoscience, Delft University of Technology, Lorentzweg 1, 2628 CJ Delft, The Netherlands}
\author{Gl\'{a}ucia Murta}\email{glauciamg.fis@gmail.com}
\affiliation{QuTech, Delft University of Technology, Lorentzweg 1, 2628 CJ Delft, The Netherlands}
\affiliation{Institut für Theoretische Physik III, Heinrich-Heine-Universität Düsseldorf, Universitätsstraße 1, D-40225 Düsseldorf, Germany}
\author{J\'{e}r\'{e}my Ribeiro}
\affiliation{QuTech, Delft University of Technology, Lorentzweg 1, 2628 CJ Delft, The Netherlands}
\affiliation{Kavli Institute of Nanoscience, Delft University of Technology, Lorentzweg 1, 2628 CJ Delft, The Netherlands}
\author{Stephanie Wehner}
\affiliation{QuTech, Delft University of Technology, Lorentzweg 1, 2628 CJ Delft, The Netherlands}
\affiliation{Kavli Institute of Nanoscience, Delft University of Technology, Lorentzweg 1, 2628 CJ Delft, The Netherlands}

\begin{abstract}
We consider the task of sharing a secret quantum state in a quantum network in a verifiable way.
We propose a protocol that achieves this task, while reducing the number of required qubits, as compared to the existing protocols. To achieve this, we combine classical encryption of the quantum secret with an existing verifiable quantum secret sharing scheme based on Calderbank-Shor-Steane quantum error correcting codes. 
In this way we obtain a verifiable hybrid secret sharing scheme for sharing qubits, which combines the benefits of quantum and classical schemes. Our scheme does not reveal any information to any group of less than half of the $n$ nodes participating in the protocol. Moreover, for sharing a one-qubit state each node needs a quantum memory to store $n$ single-qubit shares, and requires a workspace of at most $3n$ qubits in total to verify the quantum secret. Importantly, in our scheme an individual share is encoded in a single qubit, as opposed to previous schemes requiring $\Omega(\log n)$ qubits per share. 
Furthermore, we define a ramp verifiable hybrid scheme. We give explicit examples of various verifiable hybrid schemes based on existing quantum error correcting codes. 
\end{abstract}

\maketitle

\section{Introduction}

Secret sharing is a task, which allows us to securely split a secret message among $n$ network nodes, in such a way that at least a certain number of nodes is asked to collaborate in order to reconstruct the secret. 
However, one also requires that a subset with less than a certain number of nodes cannot gain any information about the secret. 
This way one can hide highly confidential and sensitive information from being exposed, for example missile launch codes or numbered bank accounts.
The splitting and sharing of the message is often performed by one designated node -- the dealer.
If the nodes do not trust the dealer, but they want a guarantee that a secret was indeed distributed, then they may wish to verify that at the end of the protocol there will be one well-defined secret that they can reconstruct. In this case, the secret sharing protocol involves an additional step of verification of the shares, and one talks about \emph{verifiable} secret sharing \cite{Chor1985,Feldman1987}.

Importantly, verifiable secret sharing is used as a subroutine for other cryptographic primitives, such as secure multipartite computation \cite{Chaum1988_unconditionally,Du2001}, byzantine agreement \cite{Feldman1997}, end-to-end auditable voting systems \cite{Ryan2015} and atomic broadcast \cite{Defago2004}.  
Likewise, a quantum analog, namely verifiable quantum secret sharing (\vqss), is a core subroutine for secure multiparty quantum computation \cite{Crepeau2002, BenOr2006} and fast quantum byzantine agreement \cite{BenOr2005}. 
Verifiable schemes, similarly to their non-verifiable counterparts, have the property that they hide information from a certain number of nodes. That is, any subset with $p$ or less nodes does not gain any information about the secret throughout the protocol. We call this property \emph{secrecy}. 

So far, many protocols have been proposed for sharing a classical secret using purely classical shares \cite{Shamir1979,Blakley1979,Krawczyk1994}, using classical and quantum shares \cite{Hillery1999,Karlsson1999,Chen2007,Xiao2004}, as well as for sharing a quantum secret with quantum shares \cite{Hillery1999,Cleve1999,Gottesman2000,Markham2008,Marin2013,Javelle2013}. This work concerns the last variant, namely schemes which share a quantum secret. Particularly, throughout this paper we will consider that the dealer shares a pure single-qubit state $\ket{\psi}$. 
In this scenario, numerous schemes for both non-verifiable quantum secret sharing \cite{Hillery1999, Cleve1999, Gottesman2000, Javelle2013, Marin2013, Bell2014} and verifiable quantum secret sharing \cite{Crepeau2002, Crepeau2005} are known.
Fundamentally, for any scheme sharing a quantum secret with only quantum resources, there exists a limit to how many nodes $p$ cannot gain any information about the secret. This limit is given by $p \leq \left\lfloor\tfrac{n-1}{2}\right\rfloor$ and can be intuitively understood as a consequence of the no-cloning theorem \cite{Wootters1982}. Indeed, if less than half of the nodes can reconstruct the secret, then there must exist at least two groups of nodes able to reconstruct it, which violates the no-cloning theorem. Moreover, if the majority of nodes recovers the secret exactly, then the remaining nodes get no information about the secret (for more details see \cite{Gottesman2000}). We will refer to schemes which saturate the above bound on $p$ as schemes with \emph{maximum secrecy}. In particular, for \vqss~with maximum secrecy, the only current construction \cite{Crepeau2002} requires that the dimension $q$ of local shares  scales with the number of nodes, $q > n$. 
Therefore, using the existing construction, we cannot find a non-trivial example of such a \vqss~scheme where the nodes hold single-qubit shares.  The reason for this scaling is that, in general, quantum secret sharing schemes are directly connected to resource-intensive quantum error correcting codes \cite{Cleve1999, Gottesman2000}. Consequently, this leads to secret sharing schemes which require $\Omega(\log n)$ of qubits per share. 

In the area of non-verifiable quantum secret sharing, some investigations have been performed to reduce the number of required qubits, particularly, by exploring ramp secret sharing schemes \cite{Ogawa2005,Marin2013} and classical encryption. 
In a ramp scheme one relaxes the constraint on the secrecy of the scheme, and therefore, allows some of the nodes to obtain partial information about the quantum state. This leads to schemes with less qubits per share. 
Additionally, the secrecy of a ramp scheme can be \emph{lifted}, i.e. the value of $p$ can be increased
by encrypting the quantum state and then sharing the encryption key via classical secret sharing, see Figure \ref{fig:secrecy}. Such a solution was dubbed hybrid secret sharing \cite{Nascimento2001,Gheorghiu2012,Fortescue2012, Singh2005}.

\begin{figure}[t]
\ffigbox{
 \includegraphics[scale=0.3]{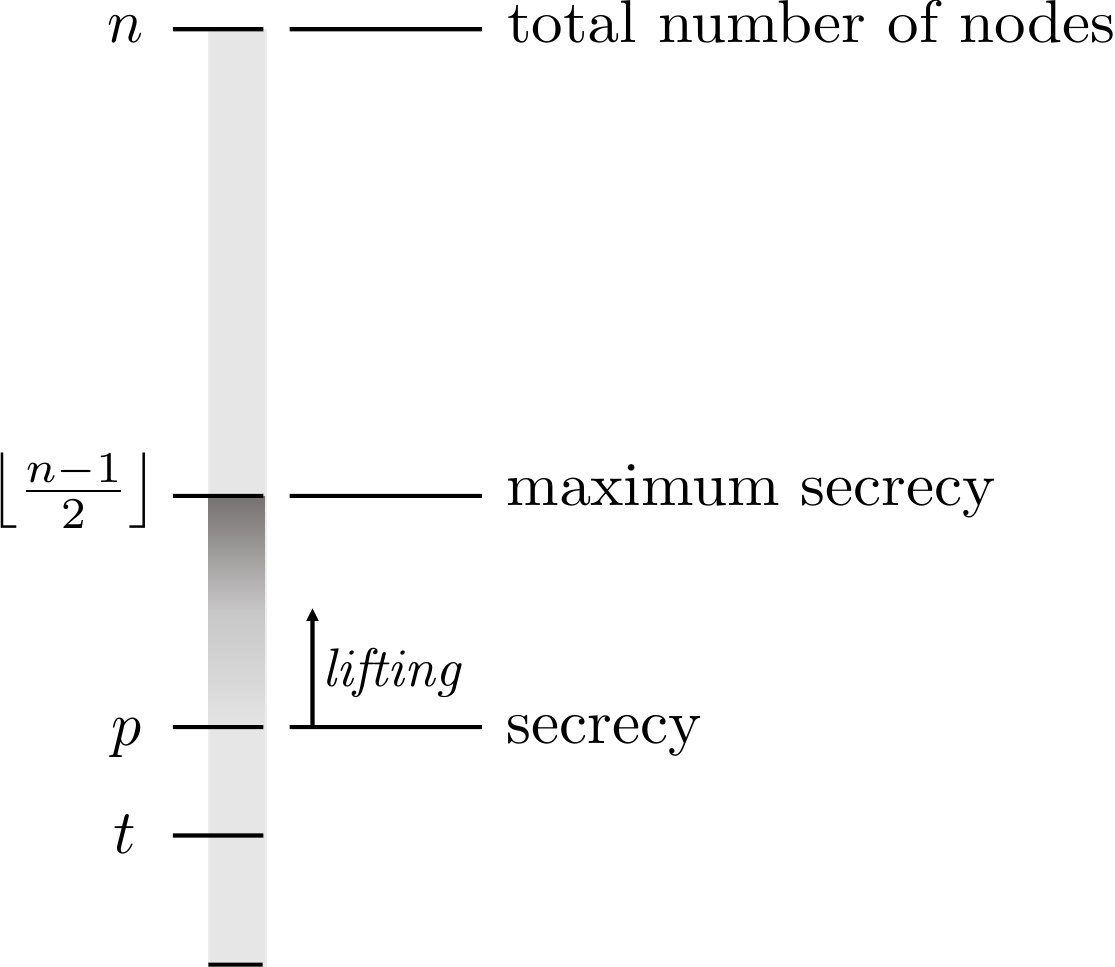}
}{
\caption{Lifting the secrecy of an $n$-node secret sharing scheme of a quantum state, i.e. increasing the value $p$ of nodes which gain no information about the secret state throughout the execution of the scheme. Here $t$ denotes the number of nodes that can perform arbitrary operations on their shares throughout the protocol, and hence corrupt the secret (active cheaters). }\label{fig:secrecy}
}
\end{figure}

In early stages of quantum network development, it would be desirable to implement \vqss~on a network with ability to control only a small number of qubits. 
Since quantum resources are expensive, a lot of effort is being put in reducing them in many areas of quantum information field, for example quantum computing or quantum simulation \cite{Bravyi2016,Steudtner2018,Moll2016,Bravyi2017,Peng2019}. However, reducing the resource requirements in the domain of distributed systems, and in particular verifiable secret sharing, has not been considered so far. Here we address the question of whether a verifiable secret sharing scheme with the maximum secrecy property (i.e. $p = \left\lfloor\tfrac{n-1}{2}\right\rfloor$) can be realized on a quantum network with less qubits. We answer this question positively by presenting a scheme which reduces quantum resources necessary for sharing a quantum secret in a verifiable way.

\section{Results}\label{sec:results}

Our contribution is three-fold. First, our scheme realizes the task of verifiable secret sharing of a quantum state using a single qubit per share. Second, we show that the protocol can be realized in a setting where each node needs to store $n$ qubits in a quantum memory and has a workspace of $3n$ qubits in total to verify the secret. For comparison, previous protocols \cite{Crepeau2002,Smith2001} require shares with $\Omega(\log n)$ qubits and each node having simultaneous control over $\Omega(r^2n\log(n))$ qubits for verification, where $r$ is the security parameter. Finally, our scheme preserves the maximum secrecy condition. This may enable qubit reductions for future implementations of cryptographic schemes, like multiparty computation or byzantine agreement, which use \vqss~as a subroutine.

We extend the idea of a hybrid scheme to verifiable quantum secret sharing. Specifically, we present a protocol that achieves the task of sharing a single-qubit quantum state $\ket{\psi}$ in a verifiable way, where the dimension $q$ of individual shares does not grow with the number of nodes $n$. In the spirit of =\cite{Nascimento2001,Gheorghiu2012,Fortescue2012,Singh2005}, we make use of classical verifiable secret sharing \cite{Rabin1989,Stinson2000} in order to obtain a verifiable hybrid scheme where each node holds at most $3n$ single-qubit shares at a time during the verification of the secret, see Outline below. Our scheme has a variety of consequences. Thanks to the classical encryption of the quantum state via quantum one-time pad \cite{Mosca2000}, our protocol can attain maximum secrecy, i.e. $p = \left\lfloor\tfrac{n-1}{2}\right\rfloor$. We show that by using a suitable classical scheme, one can beat the limit of maximum secrecy at the cost of tolerating less active cheaters (i.e. nodes that can perform arbitrary operations on their shares, see Adversary). 
Furthermore, motivated by non-verifiable schemes, we define the notion of strong threshold schemes in the context of verifiability, where any $p+1$ nodes can reconstruct the secret, any $p$ nodes do not gain any information about it, and $t$ nodes can actively cheat in the protocol. We then show that according to our definition, it is impossible to construct a verifiable strong threshold scheme. Finally, we show how to achieve a ramp hybrid scheme allowing for sharing secrets in a verifiable way. The security proof of our protocol expands on the approach suggested in \cite{Crepeau2002,Smith2001}, see Appendix \ref{app:security_vqss} for details.

\begin{figure}[b]
\ffigbox{
 \includegraphics[scale=0.3]{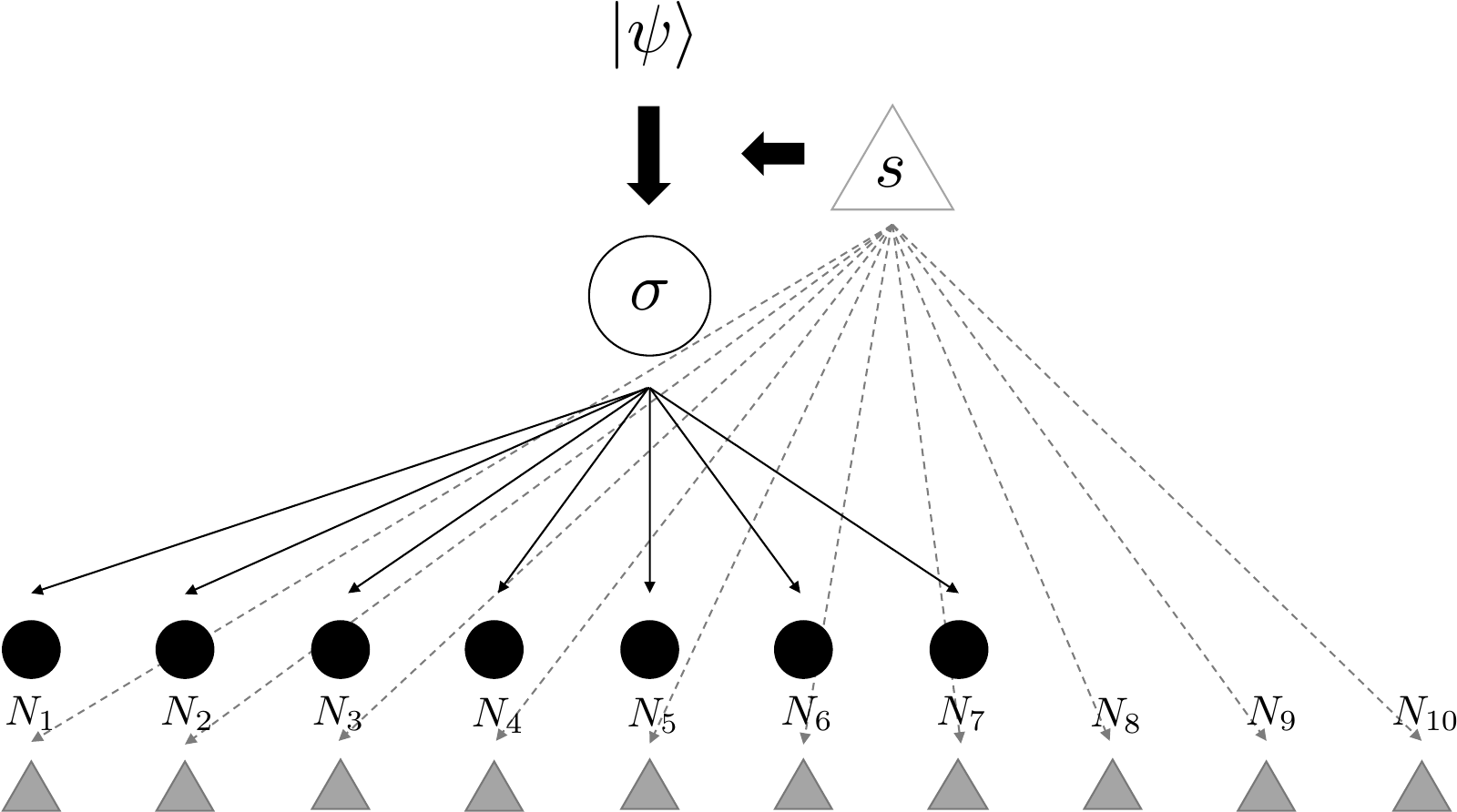}
}{
\caption{A sketch of a verifiable hybrid secret sharing (\vhss) protocol for $n = 10 $ nodes denoted $N_1,\dots,N_{10}$, with $n_q = 7$ quantum ($\CIRCLE$) and $n_c = 10$ classical ($  
	\begingroup
    \color{gray}
    \blacktriangle
  	\endgroup$) shares. The quantum secret state $\ket{\psi}$ of the dealer is encrypted using a classical key $s$. The resulting encrypted state $\sigma$ and the key $s$ are then distributed by the dealer as quantum and classical shares respectively.}\label{fig:vhss}
}
\end{figure}

\vspace{1em}

\textbf{Number of nodes.}
One key ingredient in our resource reduction is to combine quantum and classical resources in a hybrid scheme. In our model, some nodes hold quantum shares and some nodes hold classical shares. Note that nodes can have both quantum and classical shares, see Figure \ref{fig:vhss}. 
We denote the number of nodes with classical shares and the nodes with quantum shares by $n_c$ and $n_q$ respectively, and by $n$ the total number of nodes. 

\textbf{Adversary.}
We allow for the existence of $t$ malicious nodes (cheaters) in the protocol. We say that those cheaters are \emph{active}, meaning that they can perform arbitrary joint operations on their state during the execution of the protocol, in order to learn $\ket{\psi}$. We say that a protocol \emph{tolerates} $t$ active cheaters if at the end of the protocol the reconstruction of the quantum state is possible despite the presence of those cheaters. 
The nodes who follow the protocol exactly are called honest. 
We follow the common assumption that the set of malicious quantum and classical nodes is determined at the beginning of the hybrid protocol and stays fixed throughout (\emph{non-adaptive} adversary). We also assume that all nodes have access to an authenticated broadcast channel \cite{Canetti1999} and that each pair of nodes is connected by authenticated, private classical \cite{Canetti2004} and quantum \cite{Barnum2002} channels.

\begin{definition}[\hybrid]\label{def:hybrid_scheme}
A \hybrid~verifiable hybrid secret sharing scheme is an $n$-node protocol with three phases: sharing, verification and reconstruction, and two designated players, dealer $D$ and reconstructor $R$. In the sharing phase $D$ shares a pure single-qubit quantum state $\ket{\psi}$ using quantum and classical shares. In the verification phase all of the nodes verify that the set of shares defines a unique quantum state. In the reconstruction phase $R$ receives all shares from all nodes, and reconstructs the unique state defined by these shares. 
We require that the scheme satisfies the following requirements despite of the presence of $t$ non-adaptive active cheaters, except with probability exponentially small in the security parameter $r$:
\begin{itemize}[topsep=0pt,itemsep=-1ex,partopsep=1ex,parsep=1ex]
\item Soundness: if $R$ is honest and $D$ passes the verification phase, then there is a unique state $\ket{\psi}$ that can be recovered by $R$;
\item Completeness: if $D$ is honest then she always passes the verification phase. Moreover, if $R$ is also honest then the reconstructed state is exactly $D$'s state $\ket{\psi}$;
\item Secrecy: if $D$ is honest then any group of $p\geq t$ nodes cannot gain any information about the secret before reconstruction.
\end{itemize}
\end{definition}

The parameters of the scheme are determined by an underlying quantum error correcting code which we use as a building block. In particular, a relevant variable is the distance $d$ of the code. We remark that our results generalize to multi-qubit scenarios.

\begin{table*}
\centering
\begin{minipage}{\linewidth}
\renewcommand\arraystretch{1.3}
\caption{Examples of verifiable hybrid secret sharing schemes using one qubit shares coming from this work. The secret is shared among $n$ nodes. A $\{\left\lfloor\tfrac{n-1}{2}\right\rfloor,t,n\}$-\vhss~scheme uses shares from all of the nodes to reconstruct the secret, whereas $\{\left\lfloor\tfrac{n-1}{2}\right\rfloor,t,t',n\}$-ramp \vhss~scheme can reconstruct the secret without any $t'$ nodes. Both schemes tolerate $t$ active cheaters and are based on error correcting codes of \cite{Landahl2011,Bravyi1998}.}\label{tab:vhss_examples}
\begin{tabular}{|C{4cm}|C{3cm}|C{3cm}|C{3cm}|C{3cm}|}
\hline 
\multirow{2}{*}{Number of nodes $n$}   & \multicolumn{2}{c|}{$\{\left\lfloor\tfrac{n-1}{2}\right\rfloor,t,n\}$-\vhss} & \multicolumn{2}{c|}{$\{\left\lfloor\tfrac{n-1}{2}\right\rfloor,t,t',n\}$-ramp \vhss} \\ \cline{2-5}
& $t = 2$ & $t = 4$ & $t = 1$ & $t = 2$ \\
\hline \hline
 $2(t+1)^2$ & $\{8,2,18\}$ & $\{24,4,50\}$ & $\{8,1,1,18\}$ & $\{24,2,2,50\}$ \\ 
\hline 
$3t^2+3t+1 $  & $\{9,2,19\}$ & $\{30,4,61\}$ & $\{9,1,1,19\}$ & $\{30,2,2,61\}$ \\ 
\hline 
 $ 6t^2+1$  & $\{12,2,25\}$ & $\{48,4,97\}$ &$\{12,1,1,25\}$ & $\{48,2,2,97\}$ \\ 
\hline 
$ 8t^2+4t+1$  & $\{20,2,41\}$ & $\{72,4,145\}$ & $\{20,1,1,41\}$ & $\{72,2,2,145\}$ \\ 
\hline 
\end{tabular}
\end{minipage}
\end{table*}

\vspace{1em}

\subsection{\tn{\hybrid}~verifiable hybrid secret sharing protocol.} 

\begin{framed}
\noindent Outline of the verifiable hybrid secret sharing (\vhss)~protocol (see Protocol 1).

\noindent \rule{\textwidth}{0.5pt}

 1. \emph{Sharing}\\
The dealer $D$ encrypts the secret quantum state $\ket{\psi}$ using a classical key $s=ab$ and quantum one-time pad \cite{Mosca2000},
\begin{align*}
\sigma_{QS} =  \sum_{ab = \{0,1\}^2} \frac{1}{4} X^aZ^b \ketbra{\psi}_Q Z^bX^a \otimes \ketbra{ab}_{S}
\end{align*}
where $Q$ is the quantum register of the dealer and $S$ is the classical register of the encryption key.
She shares the encrypted state among the nodes using the quantum protocol and the key $s$ using the classical protocol, see Protocol 1 ``Sharing''.\\

2. \emph{Verification}\\
Nodes verify whether $D$ is honest, i.e. that the shares held by the nodes are consistent and at the end of the protocol a state will be reconstructed. For this, each node encodes the qubit received from the dealer into further $n$ qubits and sends $n-1$ of them to other nodes. Then, each node uses at most additional $2n$ ancilla qubits for one iteration of the verification procedure. There are $\mathcal{O}(r^2)$ iterations of verification, where $r$ is the security parameter. If the dealer passes the verification phase the protocol continues. Otherwise it aborts.\\

3. \emph{Reconstruction}\\
One designated node $R$ collects all shares of $\sigma$ and reconstructs it. She also reconstructs the classical key $s$ and decrypts $\ket{\psi}$. 
\\

\noindent\rule{\linewidth}{0.4pt}\\
\textbf{Remark.} Throughout the protocol each of the nodes needs to simultaneously store $n$ single-qubit shares corresponding to the encoded secret state. In the verification phase each node creates at most $2n$ ancilla qubits, performs a joint operation between these ancillas and the shares of the secret, and then measures only the ancilla qubits. This means that the nodes require a workspace of at most $3n$ qubits in total for verification.

\end{framed}

We revisit the \vqss~scheme introduced in \cite{Crepeau2002} and explore its extension to a verifiable scheme which uses single-qubit shares. The construction we use is based on Calderbank-Shor-Steane (CSS) error correcting codes \cite{Calderbank1996,Steane1996}. Then, we use the existing verifiable classical secret sharing schemes \cite{Rabin1989,Stinson2000} to combine classical encryption of the quantum secret with the \vqss~scheme to achieve an $n$-node verifiable hybrid secret sharing scheme (\vhss), see Outline. In \hybrid~the number $p$ of nodes who cannot gain any information about the quantum state is determined by the classical scheme. Moreover, $t\leq \left\lfloor\tfrac{d-1}{2}\right\rfloor$ cheaters are active and constrained by the distance $d$ of the underlying CSS code. In our scheme the secret state of the dealer $\ket{\psi}$ is encrypted using quantum one-time pad with a classical key $s$, and then both objects are shared and verified in parallel. It is, therefore, impossible to reconstruct the quantum secret without reconstructing the classical key. In the case when $n = n_q = n_c$ we achieve the following functionalities: 

\begin{itemize}
\item We construct a scheme which attains maximum secrecy using single qubit shares. Specifically, thanks to using classical encryption, we show that in our \hybrid~scheme any $p \leq \left\lfloor\tfrac{n-1}{2}\right\rfloor$ nodes coming together before reconstructing the secret, do not gain any information about it. Our \hybrid~scheme tolerates up to $t<\tfrac{n}{4}$ active cheaters. Reconstruction of the secret occurs with all of the shares.

\item We show how to achieve a \hybrid~scheme for $p >\left\lfloor\tfrac{n-1}{2}\right\rfloor$ by choosing an appropriate classical verifiable scheme \cite{Stinson2000}. In this case, however, there exists a trade-off between the number of active cheaters and secrecy, such that $n \geq p + 3t +1$. Therefore, in order to achieve higher secrecy we tolerate less active cheaters $t$. As before, reconstruction of the secret occurs with all of the shares.

\item We define a strong threshold scheme (see Definition~\ref{def:threshold}) where shares from any group of $n-t'$ nodes are sufficient for the reconstruction, no group of $p = n-t'-1$ nodes gains any information about the state. Importantly, we show that according to our definition, it is impossible to achieve a verifiable strong threshold scheme, namely, a scheme which satisfies the two above constraints and tolerates $t$ active cheaters at the same time.

\item We relax the secrecy constraint of the strong threshold scheme and construct a ramp \vhss~scheme (see Definition \ref{def:ramp}). In our ramp verifiable scheme any $n-t'$ nodes can reconstruct the secret, but any group of at most $p \leq \left\lfloor\tfrac{n-1}{2}\right\rfloor$ does not have any information about it. 
The scheme tolerates $t$ active cheaters, where $t + t' \leq \left\lfloor\tfrac{d-1}{2}\right\rfloor$ are constrained by the distance of the underlying quantum error correcting code. We denote it with \ramp.
\end{itemize} 
In the case when $n = n_c > n_q$, our \vhss~scheme allows us to construct a scheme which extends verifiable quantum secret sharing onto nodes with purely classical capabilities, see Figure \ref{fig:vhss}. That is, we use \vqss~to share a quantum secret with $n_q$ nodes, but we extend the sharing of the classical key $s$ onto $n_c > n_q$ nodes. Therefore, some of the nodes hold only classical shares but still participate in hiding of the quantum secret. Due to the properties of our protocol, this scheme can also lift the secrecy, such that no set with $p \leq \left\lfloor\tfrac{n-1}{2}\right\rfloor$ nodes can learn the quantum state before the reconstruction.

\subsection{Implications for resource reduction.}
Our scheme allows us to exploit CSS quantum error correcting codes which encode a single-qubit quantum state into single-qubit shares. Such codes are well-studied in the literature and therefore, numerous schemes with defined encoding and decoding exist \cite{Landahl2011,Bravyi1998}. In the next section we present examples of \vhss~schemes based on such codes. 
We remark that one could use approximate error correction codes and in this way increase the number of active cheaters to $2t$ \cite{Barnum2002, Crepeau2005}. However this solution requires significantly more resources, see Section \ref{sec:outlook}.

\section{Resource reduction}\label{sec:resource_reduction}

Our protocol reduces the number of qubits that need to be controlled simultaneously by each node. To do so, we adapt the protocol of \cite{Crepeau2002}, where the verification procedure requires ancillas used in parallel, to a setting where they can be used sequentially, i.e. one by one. This way, 
each node needs control over $3n$ operational qubits at a time. 
For comparison, the parallel execution of \cite{Crepeau2002} requires simultaneous control over $\Omega(r^2n\log(n))$ qubits per node, where $r$ is the security parameter.

Here we list a few examples of CSS codes leading to \vhss~schemes with single-qubit shares (also see Table \ref{tab:vhss_examples}). We express our examples in terms of maximum tolerable number of active cheaters $t$. Note that for a particular code there exists a trade-off between the number of active cheaters and the total number of nodes. 
\vspace{0.5em}

\noindent For $t=1$:
\begin{itemize}
\item  $\{3,1,7\}$-\vhss. In this scheme $n = n_c = n_q = 7$ nodes hold both quantum and classical shares. The scheme achieves maximum secrecy, i.e. no group of $p = \left\lfloor\tfrac{7-1}{2}\right\rfloor = 3$ shares acquires any information about the secret. All of the quantum shares are single-qubit shares, and each node requires control over 21 qubits at a time for the verification procedure. This example is based on the Steane's $[[7,1,3]]_2$ code, encoding 1 qubit into 7 qubits, with distance $d=3$ \cite{Steane1996}. In this scheme all shares are necessary to reconstruct the secret.

Note that the Steane's code without the classical encryption would generate a \vqss~scheme, where no $2$ nodes could gain any information about the secret. However, due to the properties of the code, a \emph{specific} group of 3 nodes could still reconstruct the secret. To compare, the existing construction to achieve a purely quantum scheme with maximum secrecy, requires individual shares of dimension $q>7$. 

\item $\{\left\lfloor\tfrac{n-1}{2}\right\rfloor,1,n\}$-\vhss. In this scheme $n_q=7$ out of $n$ nodes hold quantum single-qubit shares and $n = n_c >7 $ hold classical shares. The scheme achieves maximum secrecy. For the construction we use the Steane's $[[7,1,3]]_2$ code and a classical scheme of \cite{Rabin1989}. Therefore, in our scheme only 7 nodes need to have quantum resources, but all of the $n$ nodes can participate in verifiable secret sharing of a quantum state.
 \end{itemize}
 
\noindent For $t \geq 1$:
 \begin{itemize}
\item $\{\left\lfloor\tfrac{n-1}{2}\right\rfloor,t,n\}$-\vhss. We construct \vhss~schemes which tolerate more than one active cheater and achieve maximum secrecy. All of the nodes hold both quantum and classical shares ($n_q=n_c=n$), and the quantum shares contain a single qubit. For the construction we use higher-distance quantum error correcting codes, for example {toric} codes and color codes \cite{Landahl2011,Bravyi1998}, and \vcss~scheme of \cite{Rabin1989}. We present specific examples in Table \ref{tab:vhss_examples}. Note that each of those schemes can be expanded onto even larger total number of nodes, by using a verifiable classical secret sharing scheme with $n_c > n_q$.

\item \ramp. Based on the same higher-distance quantum error correcting codes \cite{Landahl2011,Bravyi1998}, we construct examples of ramp schemes, see Tab. \ref{tab:vhss_examples}. All of the nodes hold quantum and classical shares, however, only $n-t'$ are used to reconstruct the secret. 

\end{itemize}

\section{Methods} \label{sec:methods}

\subsection{Protocol} \label{subsec:protocol}

Our protocol is a hybrid between a classical scheme (\vcss) and a quantum scheme (\vqss) to share the classical key $s$ and the encrypted quantum state $\sigma_{QS}$, respectively. In the following we summarize the principles of these two protocols.

\subsubsection{Verifiable Classical Secret Sharing}

A verifiable classical secret sharing scheme is a scheme which shares a classical secret of the dealer among $n_c$ nodes in a verifiable way, using classical shares. The scheme is such that $p_c$ nodes cannot gain any information about the classical secret after coming together (secrecy) and there are at most $t_c$ active non-adaptive cheating nodes that the scheme tolerates. We represent the classical verifiable secret sharing protocol with a triple $(p_c,t_c,n_c)$-\vcss.
Here we treat the \vcss~scheme as a secure black box which leaks no information about the classical key $s$, even if the adversary has access to quantum side information during the execution of \vcss. 
\vcss~schemes that are information theoretically secure in the context of classical adversary have been presented in for example \cite{Rabin1989,Chaum1988_unconditionally,Stinson2000}. Here
we add it as an assumption that any \vcss~protocol used to build Protocol 1 is secure against a quantum adversary in the information-theoretic sense. 

\begin{assumption}\label{assump:vcss}
The \vcss~scheme used to build Protocol 1, does not leak any information about the secret key $s$ to any set of $p_c$ nodes, except with probability exponentially small in the security parameter $r$, even in the presence of quantum side information. That is, the scheme is information theoretically secure in the presence of a quantum adversary.
\end{assumption}

Formally, \vcss~is a classical protocol in which the dealer inputs a classical message $s$, which is shared among the nodes. 
Let $P$ be a set of size at most $p_c$, and let $\mathcal{Q}_P$ denote any quantum side information held by the nodes in set $P$ at the end of the verification phase of the \vhss. In principle, $\mathcal{Q}_P$ could be arbitrarily correlated with the classical secret key $s$. However, Assumption \ref{assump:vcss} implies that the state held by nodes in $P$ carries no information about the key $s$, other than what was known prior to the beginning of the protocol.

To the best of our knowledge, security of protocols of \cite{Chaum1988_unconditionally,Stinson2000} against an adversary with quantum side information was never formalized. We note that in Theorem~13 of \cite{Unruh2010} it was proven that any classical protocol which is statistically secure in a universal composable (UC) sense,
is also statistically UC-secure against a quantum adversary. Furthermore, \cite{Prabhakaran2013,Canetti2002} discuss the possibility of strengthening the security of \cite{Rabin1989} to UC-security. As a consequence \cite{Rabin1989} could be conjectured statistically UC-secure against a quantum adversary.

In what follows, unless specified otherwise, we will consider a classical \vcss~protocol of \cite{Rabin1989}. This scheme is secure with exponentially small probability of error $2^{-\Omega(r')}$, where $r'$ is the security parameter. Here, for convenience, we choose $r'$ such that $r'=r$, where $r$ is the security parameter of \vhss. The protocol can tolerate up to $t_c < \tfrac{n_c}{2}$ malicious nodes. In particular, it also implies that $p_c = t_c < \tfrac{n_c}{2}$.

\subsubsection{Verifiable Quantum Secret Sharing}
To construct our hybrid scheme we employ a \vqss~scheme which uses single-qubit shares. The \vqss~scheme summarized here is based on the results of \cite{Crepeau2002}.

A verifiable quantum secret sharing scheme is a scheme which shares a quantum state of the dealer among $n_q$ nodes in a verifiable way, using quantum shares. The scheme is such that $p_q$ nodes cannot gain any information about the secret (secrecy) and there are at most $t_q$ non-adaptive active cheating nodes that the scheme tolerates. We denote such a scheme with a triple $(p_q, t_q, n_q)$-\vqss.
To share a pure qubit state among $n_q$ nodes in a \vqss, the nodes agree on (an efficiently decodable) $[[n_q,1,d]]_2$ Calderbank-Shor-Steane (CSS) error correcting code $\mathcal{C}$. Such a code encodes 1 qubit into $n_q$ qubits and has distance $d$. This means that the chosen CSS code is able to correct $t_q \leq \left\lfloor\tfrac{d-1}{2}\right\rfloor$ arbitrary errors and $p_q \leq d-1$ erasure errors. 

The CSS code $\mathcal{C}$ used to perform the protocol, is defined through two binary classical linear codes, $V$ and $W$, satisfying $V^* \subseteq W$, where $V^*$ is the dual code. Then, $\mathcal{C} = V \cap \mathcal{F}W$ is a set of states of $n_q$ qubits which yield a codeword in $V$ when measured in the standard basis, and a codeword in $W$ when measured in the Fourier basis \cite{NielsenChuang2011}.
An important property of a CSS code, which is useful for the \vqss~protocol, is the fact that certain logical operations $\bar{\Lambda}$ can be implemented by applying local operations $\Lambda$ on the individual qubits held by the nodes and encoded with $\mathcal{C}$, i.e. $\bar{\Lambda}=\Lambda^{\otimes n_q}$. This property, called transversality, means that specific logical operations can be applied qubit-wise. In particular, the protocol uses the fact that \emph{(i)} applying a CNOT gate is tranversal; \emph{(ii)} applying the Fourier transform qubit-wise maps codewords of the code $\mathcal{C}$ onto codewords of the dual code $\tilde{\mathcal{C}}$; \emph{(iii)} measurements can be performed qubit-wise, but measurement outcome of every qubit must be communicated classically to obtain the result of the logical measurement.

\begin{figure*}
\centering
\includegraphics[width=0.9\textwidth]{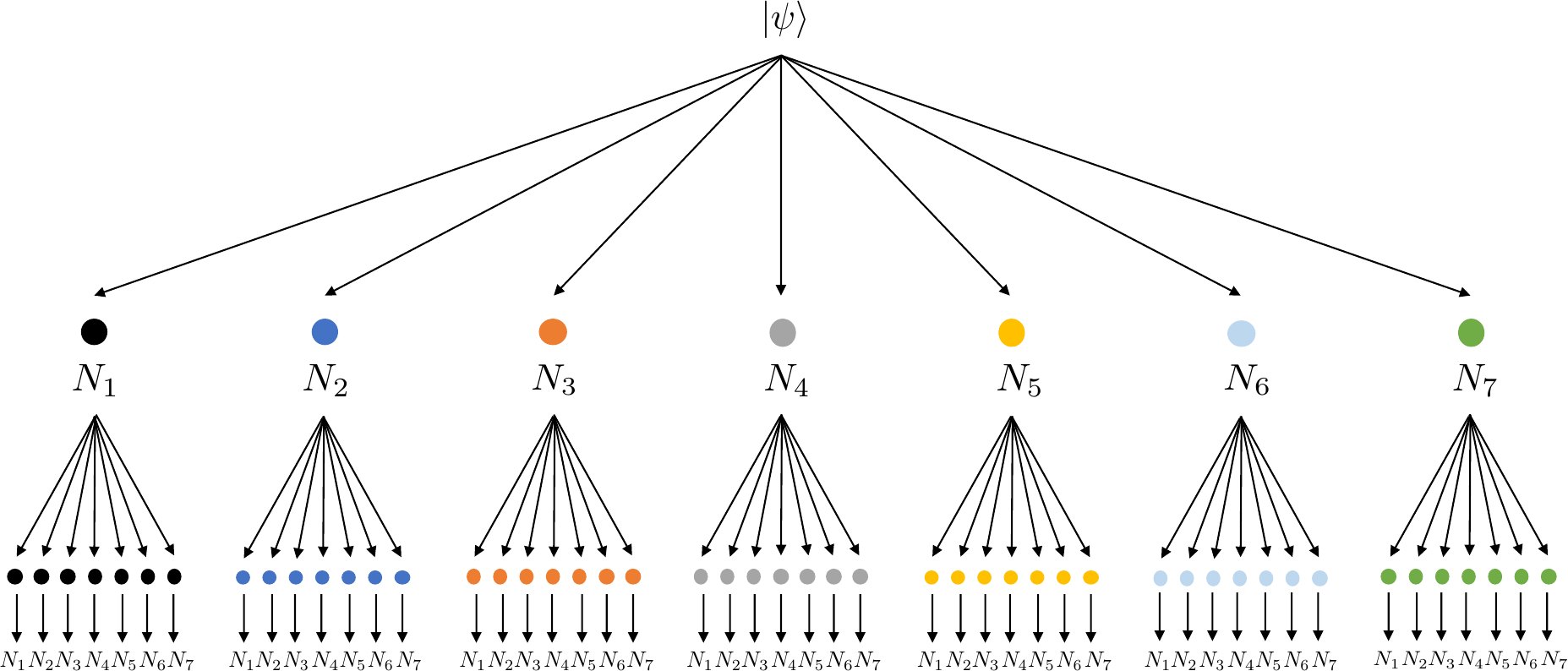}
\caption{The encoding tree for (2,1,7)-\vqss~protocol with 7 nodes $N_1,\dots,N_7$, based on the Steane's $[[7,1,3]]_2$ code. The figure represents the encoding done in the sharing phase by each of the nodes.}\label{fig:713}
\end{figure*}

In the \vqss~protocol the dealer $D$ encodes the quantum secret state $\ket{\psi}$ using the code $\mathcal{C}$ and distributes it to $n_q$ nodes. Next, each node $i$ encodes her qubit into $n_q$ further qubits and distributes those to every other node, see Figure \ref{fig:713}. This way the nodes create two levels of encoding which can be represented as a tree. 
The second level of encoding gives each node some control over all the other shares, which allows honest nodes to check consistency of all the shares. 

The protocol aims to verify whether the shares (the tree) create a codeword for which decoding is well-defined with respect to the code $\mathcal{C}$, without revealing any information about the secret state of the dealer. This property is formally defined in \cite{Crepeau2002,Smith2001} and is dubbed \good. Intuitively, a \good$_V$ tree means that for all branches of the tree which are held by honest nodes, upon measuring their shares of the tree, there exists a unique codeword in the code $V$ that can be recovered. Since $\mathcal{C} = V \cap \mathcal{F}W$, to verify that the encoded tree is \good$_\mathcal{C}$, the verification procedure first verifies that the tree is \good$_V$ when measured in the standard basis, and then that it is \good$_W$ when measured in the Fourier basis. 

We adapt the verification procedure from the work of \cite{Crepeau2002, Smith2001} to run in a sequential way. In our procedure, to verify that the encoded secret is \good$_V$ in the standard basis, the dealer and the nodes create auxiliary trees initiated in a logical $\ket{\bar{+}}$ state of the code $\mathcal{C}$. Importantly, these systems are distributed one at a time. Therefore, each node needs to control $2n$ qubits at a time: $n$ single-qubit shares for the encoded secret state, and $n$ single-qubit shares for the auxiliary $\ket{\bar{+}}$ state. We perform $r$ such checks, where $r$ is the security parameter. 

After this step, our protocol verifies that the encoded secret is \good$_W$ in the Fourier basis. To do so, the dealer and the nodes create new auxiliary trees initiated in a logical $\ket{\bar{0}}$ state of the code $\mathcal{C}$. Here an important difference is that each of the auxiliary $\ket{\bar{0}}$ states is first verified to be \good$_V$ as well, before applying the Fourier transform. This step is necessary, because one wants to make sure that the check in the Fourier basis does not introduce bit flips in the standard basis (at this point the check in standard basis for the secret state $\ket{\psi}$ has already been performed). Verifying each $\ket{\bar{0}}$ requires using extra $n$ single-qubit shares per node and is repeated $r$ times. Therefore, each node needs to control $3n$ qubits at this step: $n$ single-qubit shares for the encoded secret, $n$ single-qubit shares for a $\ket{\bar{0}}$ state, and additional $n$ single-qubit shares for the verification of $\ket{\bar{0}}$. In comparison, in \cite{Crepeau2002, Smith2001} all of the above steps are performed in parallel, and effectively, each node needs to control $\Omega(r^2n\log(n))$ at once.

In the verification phase the nodes publicly identify a set of \emph{apparent} cheaters $B$ with probability exponentially close to 1 in the security parameter $r$. Set $B$ includes all of the errors introduced by the dealer and errors introduced by the cheating nodes until the end of the verification phase. Note that there is no way to distinguish the errors introduced by the dealer and those introduced by the cheaters at this point.  The dealer will pass verification as ``honest'' if $|B| \leq t_q$. 
On the other hand, if $|B| \geq t_q$ then the protocol aborts.  

After the verification phase, the cheating nodes can still corrupt their shares. Therefore, the reconstructor $R$ runs an error correction circuit and measures syndromes, so that she can correct arbitrarily located errors introduced by the cheaters after the verification. If for a branch encoded by a particular node $i$ there have been more than $t_q$ errors, then $R$ adds that node to the set $B$ of cheaters. Otherwise, $R$ corrects errors and reconstructs branch $i$. After reconstructing all branches, she randomly picks $n - 2t_q$ shares which she has left, and reconstructs the state of the dealer. Importantly, the size of set $B$ cannot be larger than $2t_q$ at the end of the protocol. This is because the dealer $D$ and cheaters can introduce at most $t_q$ errors at the first level of encoding before verification (otherwise the protocol aborts). Before the reconstruction, the cheaters may introduce up to $t_q$ extra errors at the second level of each branch they hold. This may create extra errors at the first level, but never more than $t_q$, since the cheaters have some control over at most $t_q$ branches.

What is more, let $C_{\vqss}$ be the set of cheaters in the \vqss~and $C_{\vcss}$ the set of cheaters in \vcss.
We assume that if a node behaves maliciously in \vqss, it can also behave maliciously in \vcss, and moreover $C_{\vqss} = C_{\vcss}$. Therefore, we put $t = t_c = t_q$.
Moreover, in our \vhss~protocol we assume that the nodes have access to shared public source of randomness. This can be realized, for example, by running a classical verifiable secret sharing protocol or multipartite coin flipping. We remark that \cite{Smith2001} points out solutions to reduce the classical communication complexity of generating public randomness. In the following we will write $[1,n]$ to denote registers of nodes from $1$ to $n$.

\onecolumngrid
\small
\vspace{2em}
\begin{framed}
\noindent{Protocol 1: Verifiable Hybrid Secret Sharing (\vhss)}\\
\rule{\textwidth}{1pt}\\
\noindent Input: a qubit secret system $\ket{\psi}$ to share, CSS error correcting code $\mathcal{C} = V \cap \mathcal{F}W$.

\vspace{1em}
\noindent \textbf{SHARING}\\
\noindent \textit{Encryption}
\begin{enumerate}
\item The dealer $D$ encrypts her secret state $\ket{\psi}$ using quantum one-time pad with a classical key $s$, creating the state $\sigma_{QS}$, see Equation \eqref{eq:sigma_QS}.
\item $D$ shares the classical key $s$ among $n$ nodes using a verifiable classical secret sharing \vcss~protocol.
\end{enumerate}
\noindent \textit{Encoding}
\begin{enumerate}
\item $D$ encodes $\sigma_Q$ using $\mathcal{C}$ into ${\Phi}_{[1,n_q]}^{0,0}$, where $\sigma_Q$ is the reduced state of $\sigma_{QS}$.
\item for $i = 1,\dots,n_q$: \\
$D$ sends ${\Phi}_{i}^{0,0}$ to node $i$. \\
Each node $i$ encodes received systems using $\mathcal{C}$ into ${\Phi}_{i_{[1,n_q]}}^{0,0}$ and sends $j$-th component ${\Phi}_{i_j}^{0,0}$ to node $j$.

\end{enumerate}

\vspace{1em}
\noindent \textbf{VERIFICATION}\\ 
\textit{Z basis}

for $\ell=0$, $m = 1,\dots, r$:
\begin{enumerate}
\item $D$ prepares $\ket{\bar{+}}_{[1,n_q]}^{0,m} = \sum_{v\in V} \ket{v}$ using $\mathcal{C}$.
\item for $i = 1,\dots,n_q$: \\
$D$ sends $\ket{\bar{+}}_{i}^{0,m}$ to node $i$. \\
Each node $i$ encodes received systems using $\mathcal{C}$ into $\ket{\bar{+}}_{i_{[1,n_q]}}^{0,m}$ and sends $j$-th component $\ket{\bar{+}}_{i_j}^{0,m}$ to node $j$.
\item Nodes use shared public randomness source and get public random value $b_{0,m} \in_R \{0,1 \}$. Each node $j$:

\begin{enumerate}[label=(\alph*),topsep=0pt,itemsep=-1ex,partopsep=1ex,parsep=1ex]
\item applies the CNOT gate to her shares depending on the value of $b_{0,m}$ ($CNOT^{b_{0,m}}$). That is, for every qubit $i$, if $b_{0,m}= 0 $ the node does nothing, and if $b_{0,m}=1$ the node applies a CNOT gate with a qubit indexed by $m=0$ as a control to a qubit indexed by $m=1,\dots,r$ as a target:
\begin{align*}
\forall i=1,\dots,n_q:& \quad 
CNOT^{b_{0,m}}\left( {\Phi}_{i_j}^{0,0}, \ket{\bar{+}}_{i_j}^{0,m}  \right)
\end{align*}
\item \label{prot:broadcastZ1} measures all systems indexed $\ell = 0, ~m = 1,\dots,r$ in the $Z$ basis and broadcasts the result of the measurement. 
\end{enumerate}

\vspace{1em}

\noindent  \hspace{-3em} \textit{X basis}

\noindent \hspace{-2em} for $\ell=1,\dots,r$:

\item $D$ prepares $\ket{\bar{0}}_{[1,n_q]}^{\ell,0} = \sum_{w\in W^\perp} \ket{w}$ using $\mathcal{C}$.
\item for $i = 1,\dots,n_q$: \\
$D$ sends $\ket{\bar{0}}_{i}^{\ell,0}$ to node $i$. \\
Each node $i$ encodes received systems using $\mathcal{C}$ into $\ket{\bar{0}}_{i_{[1,n_q]}}^{\ell,0}$ and sends $j$-th component $\ket{\bar{0}}_{i_j}^{\ell,0}$ to node $j$.

\noindent \hspace{-1em} for  $m = 1,\dots, r$:

\item $D$ prepares $\ket{\bar{0}}_{[1,n_q]}^{\ell,m} = \sum_{w\in W^\perp} \ket{w}$ using $\mathcal{C}$.
\item for all $i = 1,\dots,n_q$: \\
$D$ sends $\ket{\bar{0}}_{i}^{\ell,m}$ to node $i$. \\
Each node $i$ encodes received systems using $\mathcal{C}$ into $\ket{\bar{0}}_{i_{[1,n_q]}}^{\ell,m}$ and sends $j$-th component $\ket{\bar{0}}_{i_j}^{\ell,m}$ to node $j$.

\item Nodes use shared public randomness source and get public random values $b_{\ell,m} \in_R \{0,1 \}$. Each node $j$:

\begin{enumerate}[label=(\alph*),topsep=0pt,itemsep=-1ex,partopsep=1ex,parsep=1ex]
\item applies the CNOT gate to her shares depending on the value of $b_{\ell,m}$ ($CNOT^{b_{\ell,m}}$):
\begin{align*}
\forall i=1,\dots,n_q:& \quad
CNOT^{b_{\ell,m}}\left( \ket{\bar{0}}_{i_j}^{\ell,0},\ket{\bar{0}}_{i_j}^{\ell,m} \right)
\end{align*}
\item \label{prot:broadcastZ2} measures the $m$-th system in the $Z$ basis and broadcasts the result of the measurement. 
\end{enumerate}

\item Nodes apply the Fourier transform $\mathcal{F}$ to all of their remaining shares, resulting in ${\Phi^\mathcal{F}}_{[1,n_q]_j}^{0,0}$ and $\ket{\bar{0}^\mathcal{F}}_{[1,n_q]_j}^{\ell,m}$ for each node $j$. {Note that $\ket{\bar{0}^\mathcal{F}} = \sum_{w\in W} \ket{w}$.}

\item Nodes use shared public randomness source and get public random values $b_{\ell,0} \in_R \{0,1 \}$. Each node $j$:
\begin{enumerate}[label=(\alph*),topsep=0pt,itemsep=-1ex,partopsep=1ex,parsep=1ex]
\item applies the CNOT gate to her shares depending on the value of $b_{\ell,0}$ ($CNOT^{b_{\ell,0}}$):

\begin{align*}
\forall i=1,\dots,n_q:& \quad 
CNOT^{b_{\ell,0}} \left({\Phi^\mathcal{F}}_{i_j}^{0,0}, \ket{\bar{0}^\mathcal{F}}_{i_j}^{\ell,0} \right)
\end{align*}
\item \label{prot:broadcastX} measures $\ell$-th system in the $Z$ basis and broadcasts the result of the measurement. 
\end{enumerate}

\item (Decoding leaves $Z$ basis) Broadcasted values in steps 3\ref{prot:broadcastZ1} and 8\ref{prot:broadcastZ2} yield words $\mathbf{v}_{\ell, m, i}$ from code $V$, corresponding to the second level of shares encoded by each node $i$. For each of the words, using classical decoding, the nodes:

\begin{enumerate}[label=(\alph*),topsep=0pt,itemsep=-1ex,partopsep=1ex,parsep=1ex]
\item obtain a decoded value $a_{\ell,m,i}$
\item publically check on which positions the errors have occurred, denote these positions by $B_{\ell,m,i}$. Nodes update sets $B_i = \cup_{\ell,m} B_{\ell,m,i}$ from the positions of errors which occurred in the systems encoded by node $i$. If $|B_i| > t$ then add $i$ to a global set $B$. 
\end{enumerate}  

\item (Decoding the root $Z$ basis) The nodes arrange values $a_{\ell,m,i}$ into $\mathbf{a}_{\ell,m} = \{a_{\ell,m,1},\dots, a_{\ell,m,n_q}\}$. Word $\mathbf{a}_{\ell,m}$ yields a classical codeword from the code $V$ and the nodes decode it using classical decoder of code $V$. They add the positions on which an error occurred to the global set $B$.

\item (Decoding leaves $X$ basis) Broadcasted values in step 10\ref{prot:broadcastX} yield words $\mathbf{w}_{\ell, 0, i}$ from code $W$, corresponding to the second level of shares encoded by each node $i$. For each of the words, using classical decoding, the nodes:

\begin{enumerate}[label=(\alph*),topsep=0pt,itemsep=-1ex,partopsep=1ex,parsep=1ex]
\item obtain a decoded value $a_{\ell,0,i}$
\item publically check on which positions the errors have occurred, and update sets $B_i$ and $B$ as before. Sets $B_i$ and $B$ are cumulative throughout the protocol.
\end{enumerate}  

\item (Decoding the root $X$ basis) Nodes create a codeword $\mathbf{a}_{\ell,0} = \{a_{\ell,0,1},\dots, a_{\ell,0,n_q}\}$ and decode it using classical decoder of code $W$. They add the positions on which an error occurred to the global set $B$. If $|B|>t$ then reject the dealer and abort. Otherwise continue.

\item Nodes apply an inverse Fourier transform $\mathcal{F}^{-1}$ to their remaining system and obtain global sharing of $D$ secret, i.e. each node $j$ holds ${\Phi}_{[1,n_q]_j}^{0,0}$. 

\end{enumerate}

\vspace{1em}
\noindent \textbf{RECONSTRUCTION}
\begin{enumerate}
\item Each quantum node $j=1,\dots,n_q$ sends their shares to the reconstructor $R$. Moreover, all of the $n_c$ classical nodes send their classical shares to $R$. 

\item $R$ reconstructs the classical secret key $s$ using a decoder of \vcss.

\item For each share ${\Phi}_{i_{[1,n]}}^{0,0}$ coming from encoding of node $i \notin B$, $R$ 
runs a circuit for code $\mathcal{C}$ which identifies errors. $R$ creates a set $\tilde{B}_i$ such that it contains $B_i$, $B_i \subseteq\tilde{B}_i$. If $|\tilde{B}_i|\leq t$ then errors are correctable, $R$ corrects them and decodes the $i$-th share, obtaining ${\Phi}_{i}^{0,0}$. Otherwise, $R$ adds $i$ to the global set $B$.
\item For all $i \notin B$, $R$ randomly chooses $n_q-2t$ shares ${\Phi}_{i}^{0,0}$ and applies an erasure-recovery circuit to them. $R$ obtains $\sigma_R$.

\item $R$ decrypts $\sigma_R$ using the classical key $s$ and obtains $\ket{\psi}$. 
\end{enumerate}
\end{framed}
\vspace{2em}
\normalsize
\twocolumngrid

\subsection{Security}

As discussed in previous sections, in the task of verifiable secret sharing we want to ensure that the dealer is honest and that at the end of the protocol there will be a well-defined state to be reconstructed. In this section we prove the security of Protocol 1 against $t$ non-adaptive active cheaters. First we state useful lemmas about the security of the \vqss~protocol of \cite{Crepeau2002}, which we use as a subroutine. For a detailed discussion we refer the reader to \cite{Smith2001}. We remark that we use an adapted version of \vqss~in the setting where we run the verification phase sequentially, i.e. one ancilla at a time, whereas in \cite{Crepeau2002} the verification is performed in a parallel setting, i.e. all ancillas together. In Appendix \ref{app:security_vqss} we prove that this fact does not change security statements of the original \vqss.

\begin{lemma}[soundness of \vqss]\label{lem:sound_vrqss}
In the verifiable quantum secret sharing protocol \cite{Crepeau2002}, either the honest parties hold a consistently encoded secret or dealer is caught and the protocol aborts with probability at least  $1 - 2^{-\Omega(r)}$ (see Equation \eqref{app:eq:Pabort} in Appendix \ref{app:security_vqss}).
\end{lemma} 

\begin{lemma}[completeness of \vqss]\label{lem:compl_vrqss}
In the verifiable quantum secret sharing protocol \cite{Crepeau2002}, if $D$ is honest then she passes the verification phase. Moreover, if $R$ is also honest she reconstructs $D$'s secret with probability at least $1 - 2^{-\Omega(r)}$, where $r$ is the security parameter (see Equation \eqref{app:eq:completeness} in Appendix \ref{app:security_vqss}).
\end{lemma}

Using the above lemmas we now show that our \vhss~protocol, Protocol 1, is sound and complete. 

\begin{theorem}[soundness]\label{thm:sound_vhss}
In the verifiable hybrid secret sharing protocol, Protocol 1,  either the honest parties hold a consistently encoded secret or dealer is caught and the protocol aborts with probability at least  $1 - 2^{-\Omega(r)}$.
\end{theorem} 

\begin{proof}
The soundness of the hybrid protocol is a combination of soundness statements for the \vqss~and \vcss~protocols. Formally, we need to bound the probability that one of the protocols fails,
\begin{align}\label{eq:pr_fail}
\tn{Pr}\left[ \tn{fail}_\vqss \vee \tn{fail}_\vcss  \right] \leq \tn{Pr}\left[ \tn{fail}_\vqss \right] + \tn{Pr} \left[ \tn{fail}_\vcss  \right].
\end{align}
Let us first consider $\tn{Pr} [\tn{fail}_\vcss]$. Consider the protocol of \cite{Rabin1989} whose probability of failure scales exponentially with a security parameter $r'$. We choose $r'$ such that it is equal to the security parameter of \vqss, $r'=r$, and therefore, $\tn{Pr} [\tn{fail}_\vcss] \leq 2^{-\Omega(r)}$. 

On the other hand, by Lemma~\ref{lem:sound_vrqss}, the \vqss~protocol can fail with probability $\tn{Pr} [\tn{fail}_\vqss] \leq 2^{-\Omega(r)}$. 
Therefore, we obtain
\begin{align}
\tn{Pr}\left[ \tn{fail}_\vqss \vee \tn{fail}_\vcss  \right] \leq 2^{-\Omega(r)}.
\end{align}
\end{proof}

\begin{theorem}[completeness]
In the verifiable hybrid secret sharing protocol, Protocol 1, if $D$ is honest then she passes the verification phase. Moreover, if $R$ is also honest she reconstructs $D$'s secret with probability at least $1 - 2^{-\Omega(r)}$, where $r$ is the security parameter.
\end{theorem}

\begin{proof}
For the first part of the theorem, observe that an honest dealer always passes the verification phase. Indeed, if the dealer is honest, she does not introduce any errors, neither in the \vqss, nor in the \vcss~protocol. Moreover, by the assumption that active cheaters $t$ are always bounded by the number of tolerable errors, the \vhss~protocol can always correct the arising errors and the verification phase always accepts an honest dealer. 

For the second part of the theorem, as in the soundness statement, we calculate the probability that the \vhss~protocol fails with an honest dealer, 
\begin{align}
\tn{Pr}\left[ \tn{fail}_\vqss' \vee \tn{fail}_\vcss'  \right]\leq \tn{Pr}\left[ \tn{fail}_\vqss' \right] + \tn{Pr} \left[ \tn{fail}_\vcss'  \right].
\end{align}
For the classical \vcss~protocol, as before, we consider the protocol of \cite{Rabin1989}. By choosing the security parameter of the classical protocol such that $r'=r$, we obtain $\tn{Pr} [\tn{fail}_\vcss']\leq 2^{-\Omega(r)}$.
For the \vqss~protocol, if $R$ is also honest, by Lemma \ref{lem:compl_vrqss} the probability that the verification phase fails to identify the set $B$ of apparent malicious nodes, occurs with probability $2^{-\Omega(r)}$, see Appendix \ref{app:security_vqss} for details. Therefore,
\begin{align}
\tn{Pr}\left[ \tn{fail}_\vqss' \vee \tn{fail}_\vcss'  \right]\leq 2^{-\Omega(r)}.
\end{align}
\end{proof}

The encryption of the secret with a classical key has significant consequences for the secrecy of the \vhss~scheme. We expand on it in the theorem below. Note that in a \vqss~\cite{Crepeau2002} the secrecy property holds for any $p_q \leq 2t_q$ nodes not being able to learn any information about the dealer's secret. However, in our \vhss~scheme we choose a classical scheme such that $p_c = p > 2t_q$, and therefore, we lift the secrecy of the \vqss~scheme (for a detailed discussion see Sec.~\ref{subsubsec:private_vhss} below).

\begin{theorem}[secrecy]\label{thm:privacy}
In the verifiable hybrid secret sharing protocol, Protocol 1, when $D$ is honest and there is at most $t$ active cheaters in the verification phase, no group of at most $p=p_c$ nodes learns anything about $D$'s secret state throughout the protocol, where $p_c$ is the secrecy of the underlying classical scheme, except with probability exponentially small in the security parameter $r$.
\end{theorem}

\begin{proof}
The state describing the dealer's encrypted quantum secret and the randomly chosen classical encryption key $s=ab$ is
\begin{align}\label{eq:sigma_QS}
\sigma_{QS} =  \sum_{ab = \{0,1\}^2} \frac{1}{4} X^aZ^b \ketbra{\psi}_Q Z^bX^a \otimes \ketbra{ab}_{S}
\end{align}
where $Q$ is the quantum register of the dealer and $S$ is the classical register of the encryption key.
By Assumption \ref{assump:vcss} the classical \vcss~scheme is secure and does not leak any information about the key $s=ab$ to any set of $p_c$ nodes, even in the presence of a quantum adversary, except with probability exponentially small in the security parameter $r$. Therefore, without the knowledge of the encryption key $s$, the quantum state shared by the dealer as seen by the rest of the nodes is maximally mixed,
\begin{align}
\begin{split}\label{eq:sigma_D}
\sigma_Q & = \tr_{S}(\sigma_{QS}) = \\
& =\sum_{ab = \{0,1\}^2} \frac{1}{4} X^aZ^b \ketbra{\psi}_Q Z^bX^a  = \frac{\mathds{1}_Q}{2}.
\end{split}
\end{align}
Before sending out the shares, the dealer applies an encoding $\mathcal{E}_{Q}$ to the quantum register $Q$, so that 
\begin{align}
\forall \ket{\psi} \quad 
\tr_{S}((\mathcal{E}_{Q}\otimes \mathds{1}_S) (\sigma_{QS}))& = 
\mathcal{E}_{Q} (\tr_{S}(\sigma_{QS})) \\ & = \mathcal{E}_{Q} (\sigma_Q) =:
\rho_{[1, n_q]},
\end{align}
where $\rho_{[1, n_q]}$ is an $n_q$-qubit state sent by the dealer to $n_q$ nodes. Importantly, since $\mathcal{E}_{Q}$ and $\sigma_Q$, Equation \eqref{eq:sigma_D}, are independent of $\ket{\psi}$, $\rho_{[1, n_q]}$ is also independent of $\ket{\psi}$. Subsequently, the honest nodes do their encoding $\mathcal{E}$, and the malicious nodes can perform any (CPTP) operation $\mathcal{A}$ that they desire. After this step, since $\mathcal{E}$ and $\mathcal{A}$ do not depend on $\ket{\psi}$, the state of the $n_q$ nodes $\rho_{[1, n_q]}'$ is independent of $\ket{\psi}$. In the classical scheme any group of $p_c$ or fewer nodes has no information about $s$. Hence, the partial state of any $p = p_c$ or fewer nodes in \vhss~does not depend on $\ket{\psi}$ and no information about the dealer's secret can be obtained, except with probability exponentially small in $r$.

\end{proof}

\subsection{Verifiable Hybrid Schemes}\label{sec:ver_hybrid_schemes}

Our protocol for \vhss, Protocol 1, leads to a variety of schemes, depending on the parameters of the underlying \vqss~and \vcss~protocols. In this section we discuss the trade-offs between those parameters and specify what schemes can be achieved with our protocol.

\subsubsection{Verifiable schemes with maximum secrecy}\label{subsubsec:private_vhss}

In any \vqss~scheme based on an error correcting code with distance $d$, any group of at most $d-1$ nodes cannot recover information about the secret. As mentioned before, this is due to the fact that a code of distance $d$ can correct up to $d-1$ erasures, and therefore any $n-(d-1)$ nodes can recover the state perfectly. In particular, it implies that $d-1$ nodes do not have any information about the encoded state \cite{Gottesman2000}. Quantum Singleton bound \cite{Knill2000} allows that $n \leq 2d-1$ for codes encoding a single qubit. 
The construction of \cite{Crepeau2002} saturates this inequality, and therefore allows for attaining $p = \left\lfloor\tfrac{n-1}{2}\right\rfloor$, which we refer to as maximum secrecy. However, this construction uses systems of local dimension $q > n$ and is based on quantum Reed-Solomon codes \cite{Aharonov1997}. 

To remedy this problem, we use a \vqss~scheme based on CSS codes with single-qubit shares, at the cost of reducing secrecy. However, in our \vhss~scheme, we combine this with a classical scheme for which $p_c>2t_q$. Specifically, the \vcss~protocol of \cite{Rabin1989} tolerates up to $\left\lfloor\tfrac{n-1}{2}\right\rfloor$ cheaters. This allows us to maximally lift the secrecy of the quantum scheme to the one attainable by the \vqss~of \cite{Crepeau2002}. 

\begin{lemma}[\vhss~with maximum secrecy]\label{lem:max_private_vhss}
Given a $[[n,1,d]]_2$ CSS error correcting code and a \vcss~scheme tolerating up to $ \left\lfloor\tfrac{n-1}{2}\right\rfloor$ classical active cheaters, Protocol 1 provides a way to construct a $\{\left\lfloor\tfrac{n-1}{2}\right\rfloor,t,n\}$-\vhss~scheme with maximum secrecy $p = \left\lfloor\tfrac{n-1}{2}\right\rfloor$, tolerating $t \leq \left\lfloor\tfrac{d-1}{2}\right\rfloor$ active cheaters, where all of the shares are used to recover the quantum secret state. 
\end{lemma}

Furthermore, we can explore other classical verifiable schemes in the context of lifting secrecy in \vhss. In~\cite{Stinson2000}  a classical \vcss~scheme was presented, which has a strong secrecy property: any $p_c > t_c$ nodes cannot learn any information about the classical secret (for comparison, in the scheme of \cite{Rabin1989} $p_c = t_c$). However, this scheme is able to tolerate up to $t_c \leq \left\lfloor\tfrac{n_c-1}{4}\right\rfloor$ active classical cheaters. Additionally, there exists a trade-off between the number of nodes $n$, and the numbers of cheaters, i.e. $n_c \geq p_c+3t_c+1$ (for details see Section 3.2 of \cite{Stinson2000}). Consequently, this allows us to construct a \vhss~scheme lifting the secrecy beyond $\tfrac{n}{2}$, but at the cost of tolerating less active cheaters $t$. Note that the classical scheme was proven to be information theoretically secure against a classical adversary, and by Assumption \ref{assump:vcss} we assume it remains information theoretically secure against quantum adversary. Moreover, the protocol was shown to be perfectly secure, i.e. with zero probability of error. Therefore, secrecy achieved in a \vhss~which uses this protocol as a subroutine, is exact and does not depend on the security parameter $r$.

\begin{lemma}\label{lem:more_private_vhss}
Given a $[[n,1,d]]_2$ CSS error correcting code and a \vcss~scheme with $n \geq p + 3t +1$, Protocol 1 provides a way to construct a $\{p,t,n\}$-\vhss~scheme. In particular, to achieve $p > \left\lfloor\tfrac{n-1}{2}\right\rfloor$ the scheme tolerates $t < \frac{1}{3}\left( n-p - 1\right)$ active cheaters. All of the shares are used to recover the quantum secret state. 
\end{lemma}

\subsubsection{Threshold verifiable schemes}

In the literature of secret sharing schemes, one often considers schemes which have a property called \emph{threshold} \cite{Blakley1979, Shamir1979}. This property can be stated as the requirement that there exists $p>0$, such that no subset of less than $p$ shares reveals any information about the state of the dealer, while any subset of $p+1$ shares allows to perfectly reconstruct the state. 
Importantly, in such schemes, there are no actively cheating nodes in the protocol. 

Since in Protocol 1 we allow for the existence of active cheaters, let us consider a definition of a threshold scheme when there are $t>0$ active cheaters. We will call it a strong threshold scheme. In this case, in the reconstruction phase the reconstructor $R$ receives shares from $p+1 = n-t'$ of the nodes. Among those, up to $t$ of them can be arbitrarily corrupted.

\begin{definition}[strong threshold scheme]\label{def:threshold}
A strong threshold (verifiable) secret sharing scheme is a scheme where:
\begin{enumerate}[topsep=0pt,itemsep=-1ex,partopsep=1ex,parsep=1ex]
 \item \label{point_1}Any set of shares held by $p = n-t'-1$ nodes does not reveal any information about the secret state. 
 \item \label{point_2}The reconstructor is able to perfectly reconstruct the secret state with the set of shares from any $n - t'$ nodes. 
 \end{enumerate}
 The above conditions hold in the presence of $t>0$ active cheaters.
\end{definition}

In the literature of classical verifiable secret sharing  a similar definition of threshold is satisfied in the presence of cheaters. For example, the scheme of \cite{Schoenmakers1999} considers a situation when honest shares are flagged. Therefore, the reconstructor knows which $n-t'$ honest shares to pick for the reconstruction. However, in our case, the reconstructor \emph{does not} know which shares are honest and which are not. In such a situation, this definition cannot be satisfied, which we show in the following proposition.

\begin{proposition}
It is impossible to construct a strong threshold secret sharing scheme according to Definition~\ref{def:threshold}.
\end{proposition}

\begin{proof}
From point \ref{point_2} of Definition \ref{def:threshold} we have that $R$ must be able to reconstruct the secret state from any $n-t'$ shares, in particular, 
she must be able to do so when receiving $n-t'-t$ honest shares and $t$ arbitrary ones. This implies that she is able to recover the state from the $n-t'-t$ honest shares alone. On the other hand, from point \ref{point_1} of Definition \ref{def:threshold} no $n-t'-1$ shares reveal any information, which implies that we must have $n-t'-t> n-t'-1$. The only way to satisfy this inequality is when $t=0$. 
\end{proof}

\noindent \textbf{Remark.} Similarly to  \cite{Schoenmakers1999}, it is possible to add a flagging system to Protocol 1 using  techniques from \cite{Barnum2002, Crepeau2005}. Indeed, there, one uses a quantum authentication scheme to flag whether the shares are honest or not. However, as mentioned before, this happens at a significant qubit cost. Since our objective is to reduce the number of qubits, we explore a alternative direction in the next section.

\subsubsection{Ramp verifiable schemes}

In the previous section, we have seen that it is impossible to construct a strong threshold scheme which tolerates active cheaters according to Definition \ref{def:threshold}. In particular, this result also applies to verifiable schemes. Therefore, here we allow for a gap between the number of nodes $p$ that obtain no information about the secret and the number of nodes $n-t'$ necessary to reconstruct the secret, and we introduce a definition of a ramp verifiable scheme.

\begin{definition}\label{def:ramp}
A ramp verifiable secret sharing scheme is a scheme where any $n - t'$ nodes can reconstruct the secret, but any $p$ nodes cannot gain any information about the secret state, for some $p < n-t'-1$. The scheme can verify the dealer in the presence of $t$ active cheaters. We denote such a scheme with $\{p,t,t',n\}$-ramp.
\end{definition}

Relating to discussion in Section~\ref{subsubsec:private_vhss}, we see that the purely quantum \vqss~scheme of \cite{Crepeau2002} allows for constructing a ramp scheme with secrecy $p \leq \left\lfloor\tfrac{n-1}{2}\right\rfloor$. However, for qubit CSS codes this equality is not saturated. Therefore, as before we use a classical scheme \cite{Rabin1989} to increase the value of $p$ (lift the secrecy) as compared to the purely quantum ramp scheme. We obtain the following result.

\begin{lemma}[Ramp \vhss]
Given a $[[n,1,d]]_2$ CSS error correcting code and a \vcss~scheme tolerating up to $\left\lfloor\tfrac{n-1}{2}\right\rfloor$ classical active cheaters, Protocol 1 provides a way to construct a \ramp~scheme with $p = \left\lfloor\tfrac{n-1}{2}\right\rfloor$, where the quantum state can be recovered with shares from any $n-t'$ nodes in the presence of $t$ active cheaters, and $t + t' \leq \left\lfloor\tfrac{d-1}{2}\right\rfloor$.
\end{lemma}

By putting $t'=0$ we require reconstruction with all of the shares and recover the result of Lemma \ref{lem:max_private_vhss}. Note that if we are interested in maximizing the number of cheaters and minimizing the number of the shares necessary for reconstruction, we can put $t = t' =  \left\lfloor\tfrac{d-1}{4}\right\rfloor$.

\section{Outlook}\label{sec:outlook}
We presented a protocol which achieves the task of sharing a quantum secret in a verifiable way, which reduces the number of qubits necessary to realize the protocol. In our scheme each node requires an $n$-qubit quantum memory and a workspace of at most $3n$ qubits in total. By combining classical encryption with a quantum scheme we showed that we can construct a variety of verifiable hybrid schemes attaining maximum secrecy. We proved that our protocol is secure in the presence of active non-adaptive adversary. 

We remark that there is a dependence between the number of cheaters tolerated by a verifiable secret sharing protocol and quantum resources necessary to realize it. 
The number of cheaters can be increased to $2t$ by using approximate quantum error correction based on quantum authentication schemes \cite{Barnum2002, Crepeau2005}. Indeed, in \cite{BenOr2006} the authors showed that by employing quantum authentication techniques, the \vqss~scheme of \cite{Crepeau2002} can tolerate up to $\frac{n}{2}$ malicious nodes. 
In this case, the power of the verification scheme increases up to the number of tolerable erasures of the code, and one can effectively tolerate twice as many malicious nodes. However, authentication schemes typically require another level of error correction, where the size of the code scales exponentially in the security parameter of the authentication. Therefore, such schemes increase the number of qubits required to realize the protocol.

\section{Acknowledgments}
We thank J. Helsen, B. Dirkse and P. Mazurek for valuable discussions and insights. This work was supported by STW Netherlands, NWO VIDI grant, ERC Starting grant and NWO Zwaartekracht QSC. GM was also funded by the Deutsche Forschungsgemeinschaft (DFG, German Research Foundation) under Germany's Excellence Strategy – Cluster of Excellence 
Matter and Light for Quantum Computing (ML4Q) EXC 2004/1 – 390534769.


%

\appendix

\section{Security of the \vqss~scheme \cite{Crepeau2002} in the sequential setting} \label{app:security_vqss}

\begin{proof}[Proof of Lemma~\ref{lem:sound_vrqss}]

Here we state the soundness of the \vqss~protocol. Since we use the \vqss~in the sequential setting instead of the original parallel one, we restate security in the sequential setting. Our techniques are inspired by the approach suggested in  \cite{Crepeau2002,Smith2001}.

To prove the soundness of the \vqss~protocol, we bound the probability that the state held by the nodes after the verification phase is close to a codeword in $\mathcal{C} = V \cap \mathcal{F} W$ with at most $t$ errors on the first level of encoding in the verification phase, or that the protocol aborts, and therefore, the dealer is caught. $V$ denotes a space spanned by  $\{\ket{v}: v \in V^C\}$, where $V^C$ is a classical code space. Similarly, $\mathcal{F} W$ is spanned by $\{\mathcal{F}\ket{w}: w\in W^C\}$, where $\mathcal{F}$ is the Fourier transform and $W^C$ is a classical code space such that the dual code $V^{C*} \subseteq W^C$. 

Recall that in the protocol we encode the secret of the dealer into two levels of encoding. We will argue that performing verification on the second level of encoding is equivalent to verification on the first level of encoding. If a state is encoded once using $\mathcal{C}$, and has at most $t$ errors, then the encoding defines a unique state. Therefore, it is enough to count the number of errors  present in the first level of encoding and verify that there are at most $t$. However, the protocol requires two levels of encoding to make sure that no node has complete control over all shares. This implies that we cannot perform the verification directly at the first level. But since all the operations we use for verification are  (essentially) transversal for code $\mathcal{C}$, we can argue about the verification as if it was performed on the first level.

In order to check for errors, it is enough to check for errors in the $Z$ basis and errors in the $X$ basis. 
Let $V_t$ be the space of words that have at most $t$ errors in the $Z$ basis as compared to a codeword in $V$. In particular, if one measures a state $\ket{v} \in V_t$ in the $Z$ basis, the outcome is a word in the space $V_t^C$, where $V_t^{C}$ is the space of strings having at most $t$ compared to a string in the classical code $V^C$. Similarly, we can define $(\mathcal{F} W)_t$ as the space of words that have at most $t$ errors in the $X$ basis as compared to a codeword in $W$. This means that if one measures a state $\ket{w} \in (\mathcal{F}W)_t$ in the $X$ basis, the outcome is a word in the space $W_t^C$, where $W_t^{C}$ is the space of strings having at most $t$ compared to a string in the classical code $W^C$.

Considering the above argument, now we proceed with proving soundness of verification of the state in the $Z$ basis and as if we were considering only one level of encoding.

Without loss of generality, we can decompose the state of the nodes after the sharing phase in spaces $V_t$ and $V_t^\perp$,
\begin{align}\label{app:eq:psi_sh}
\rho_{sh} = \sum_{i}q_i \ketbra{\psi_i},
\end{align}
with 
$
\ket{\psi_i} = a_i \ket{\tilde{\psi}_i} + b_i \ket{\tilde{\psi}_i^\perp},
$
where $\ket{\tilde{\psi}_i} \in V_t$ and $\ket{\tilde{\psi}_i^\perp} \in V_t^\perp$.
In words, the state after the sharing phase is a mixture of pure states which have components in $V_t$ and $V_t^\perp$. 

Moreover, let $\rho_{ver(Z)}$ be the state of all the nodes after the verification phase in the $Z$ basis.

We will show that 
\begin{center}
``conditioned on not aborting, the state $\rho_{ver(Z)}$ is close to a codeword in the space $V_t$ or the verification phase aborts with high probability''.
\end{center} 
By definition of the space $V_t$, $\rho_{ver(Z)}$ belongs to $V_t$, if by measuring it in the $Z$ basis one obtains with certainty an outcome corresponding to a string $v\in V_t^C$. Therefore, we will quantify ``the state $\rho_{ver(Z)}$ is close to a codeword in the space $V_t$'' with a high probability of getting an outcome $v\in V_t^C$ when measuring $\rho_{ver(Z)}$. Alternatively, one can think of a situation in which first a measurement on the initial state is performed and then the verification takes place. To prove the security statement we will use a tool called ``quantum-to-classical'' reduction, which relates the statistics obtained in the two situations. That is, in order to compute the probability of aborting in the verification phase of the \vqss~protocol or the probability that the resulting state is in $V\cap \mathcal{F}W$, we will analyze the situation in which the state is  measured \emph{before} the verification.

\vspace{1em}
\textbf{Probability of aborting.} In order to evaluate probability of aborting, we will follow the solution suggested in \cite{Smith2001} for the parallel execution of the \vqss~and we will show how to use this result for the sequential setting. To do so, let us fix a round $(0,m)$, with $m>0$. For this round we can use the ``quantum-to-classical'' reduction. It states that the two following situations are equivalent: \emph{(i)} the honest nodes measure their shares of $\rho_{ver(Z)}$ in the standard basis at the end of the verification phase; \emph{(ii)} the honest nodes measure their shares of $\rho_{sh}$ and an $m$-th ancilla right after they have been distributed, i.e. before running the verification of round $(0,m)$. 
Formally, 
\begin{align} \label{app:eq:MCNOT}
\forall m ~   \mathcal{M}_0 \mathcal{M}_m CNOT_{0,m}^{b_{0,m}} = \mathcal{M}_m CNOT_{0,m}^{b_{0,m}}\mathcal{M}_m \mathcal{M}_0
\end{align}
where $\mathcal{M}_0$ and $\mathcal{M}_m$ denote measurements of the state of the nodes and $m$-th ancilla respectively. $CNOT_{0,m}^{b_{0,m}}$ denotes a CNOT gate performed with $\rho_{sh}$ as a control and the $m$-th ancilla as target. 
Note that if the nodes perform measurements right after the shares are distributed (situation \emph{(ii)}) they only need to handle classical data from that moment on. Therefore, ``quantum-to-classical'' reduction means that the verification phase of the quantum \vqss~protocol ($Q$-protocol) can be reduced to a corresponding verification in a classical protocol ($C$-protocol). 
That is to say, measurement outcomes in $Q$-protocol and $C$-protocol are exactly the same and the moment when the measurement is performed does not change the behavior of the protocol. Since the measurement is performed in the standard basis and the CNOT gate acts as a bit flip in the standard basis, the two operations commute. 

Let us look now at the sequential execution of $Q$-protocol and $C$-protocol. Expanding the above dependence onto $m$ sequential rounds, we obtain
\begin{align} \label{app:eq:q-to-c}
\begin{split}
& \mathcal{M}_0 \mathcal{M}_r CNOT_{0,r}^{b_{0,r}} \dots \mathcal{M}_1 CNOT_{0,1}^{b_{0,1}} = \\
 & = \mathcal{M}_r CNOT_{0,r}^{b_{0,r}}\mathcal{M}_r \dots \mathcal{M}_1 CNOT_{0,1}^{b_{0,1}} \mathcal{M}_1 \mathcal{M}_0
\end{split}
\end{align}
In particular, this means that the probability of aborting in the sequential $Q$-protocol can be reduced to considering the probability of aborting in the sequential $C$-protocol, 
\begin{align}\label{app:eq:abortQ_abortC}
\tn{Pr}[ \neg \tn{abort}_Q] =  \tn{Pr}[ \neg \tn{abort}_C].
\end{align} 

Consider the corresponding $C$-protocol for round $(\ell = 0,m)$: the nodes have \emph{classical} bit strings $v_{0,0}$ and $v_{0,m}$. They wish to verify whether $v_{0,0}$ is a string in the space $V_t^{C}$. To do so the (honest) nodes compute bit-wise $v_{0,m} + b_{0,m}v_{0,0}$ according to public random bit $b_{0,m}$. They broadcast the result and create the set of apparent cheaters $B$. 

In the $C$-protocol, the string $v_{0,0}$ can either be a string in $V_t^C$ or not. This depends on the shared state \eqref{app:eq:psi_sh}, and therefore happens with probabilities 
\begin{align}
\tn{Pr}[v_{0,0} \in V_t^C] =  \sum_i q_i |a_i|^2 =:a, \\
\tn{Pr}[ v_{0,0} \notin V_t^C] = \sum_i q_i |b_i|^2 =:b,
\end{align}
respectively. Indeed, the probability that any of the $\ket{\psi_i}$ from \eqref{app:eq:psi_sh} yields a string from $V_t^C$ (resp.~not in $V_t^C$) is given by $|a_i|^2$ (resp.~$|b_i|^2$). 
In the case when $v_{0,0}$ is a string in $V_t^C$, the verification always passes and we have that $\tn{Pr}[\neg \tn{abort}_C | v_{0,0} \in V_t^C] = 1$. On the other hand, if $v_{0,0}$ is not a string in $V_t^C$, then for all bit strings $v_{0,m}$ there exists at most one bit $b_{0,m}$ such that $v_{0,m} + b_{0,m}v_{0,0}$ is a string in $V_t^C$. Since $b_{0,m}$ is chosen independently of $v_{0,m}$ and $v_{0,0}$, and uniformly at random, the probability that $v_{0,m} + b_{0,m}v_{0,0}$ a codeword is at most $\tfrac{1}{2}$. Since the above is true for any value of $v_{0,m}$, in particular it must be true even if $v_{0,m}$ depends on the previous rounds $1,\dots,m-1$. Therefore, the overall probability $p$ that the verification phase of the $C$-protocol does not abort given that $v_{0,0}$ is not a string in $V_t^C$, is at most 
\begin{align} \label{app:eq:p_bound}
p = \tn{Pr}[\neg \tn{abort}_C | v_{0,0} \notin V_t^C] \leq 2^{-r}.
\end{align}
The above consideration allows us to write that the probability of the $C$-protocol not aborting is 
\begin{align}
\begin{split}
\tn{Pr}[ \neg \tn{abort}_C] & = \tn{Pr}[v_{0,0} \in V_t^C] \tn{Pr}[\neg \tn{abort}_C | v_{0,0} \in V_t^C] \\
& ~+ \tn{Pr}[v_{0,0} \notin V_t^C] \tn{Pr}[\neg \tn{abort}_C | v_{0,0} \notin V_t^C].
\end{split}
\end{align} 
Since $ \tn{Pr}[ \neg \tn{abort}_Q] =  \tn{Pr}[ \neg \tn{abort}_C]$, Equation \eqref{app:eq:abortQ_abortC}, in the $Q$-protocol we have
\begin{align} \label{app:eq:a'b'}
 \tn{Pr} \left[\neg \tn{abort}_Q \right] = a + pb.
\end{align}

\vspace{1em}
\textbf{Probability of measuring a string in $V_t^C$.} 

Now our objective is to evaluate $\tn{Pr} \left[v_{0,0} \in V_t^C | \neg \tn{abort}_Q \right]$.
By ``quantum-to-classical'' reduction argument \eqref{app:eq:q-to-c}, we know that the $C$-protocol should yield the same statistics as the $Q$-protocol, 
\begin{align}
\tn{Pr} \left[v_{0,0} \in V_t^C | \neg \tn{abort}_Q \right] = \tn{Pr} \left[v_{0,0} \in V_t^C | \neg \tn{abort}_C \right].
\end{align}
From the considerations about the probability of aborting, using the rules of probability, we can compute 
\begin{align}\label{app:eq:F_P_a'b'}
\tn{Pr} \left[v_{0,0} \in V_t^C| \neg \tn{abort}_Q \right] =
 \frac{a}{a + pb }.
\end{align}

\vspace{1em}
Now let us combine the statements about probability of aborting and probability of measuring a string in $V_t^C$. Using the ``quantum-to-classical'' reduction, we can formally reformulate the initial statement ``conditioned on not aborting, the state $\rho_{ver(Z)}$ is close to a codeword in the space $V_t$, or the verification phase aborts with high probability'' as
\begin{align}\label{app:eq:Pr_V}
\begin{cases}
\tn{Pr} \left[v_{0,0} \in V_t^C| \neg \tn{abort}_Q \right] > 1-\delta \\
\text{or}\\
\tn{Pr} \left[v_{0,0} \in V_t^C| \neg \tn{abort}_Q \right] \leq 1-\delta \\
\qquad \text{ and } \Pr[\tn{abort}_Q] \geq 1 - \frac{2^{-r}}{\delta}
\end{cases}
\end{align}
where $\delta$ is a threshold for probability of measuring a string from $V_t^C$. 
Indeed, using equations \eqref{app:eq:a'b'} and \eqref{app:eq:F_P_a'b'} we can express $\tn{Pr} \left[v_{0,0} \in V_t^C| \neg \tn{abort}_Q \right]$ as a function of $\tn{Pr}[ \neg \tn{abort}_Q]$,
\begin{align}\label{app:eq:F_P}
\tn{Pr} \left[v_{0,0} \in V_t^C| \neg \tn{abort}_Q \right] = \frac{\tn{Pr}[ \neg \tn{abort}_Q] - p}{\tn{Pr}[ \neg \tn{abort}_Q] (1-p)}
\end{align}
Now, either $\tn{Pr} \left[v_{0,0} \in V_t^C| \neg \tn{abort}_Q \right] > 1-\delta$ and the first condition is satisfied, or $\tn{Pr} \left[v_{0,0} \in V_t^C| \neg \tn{abort}_Q \right] \leq 1-\delta$ and using \eqref{app:eq:F_P} we 
get 
\begin{align}
\Pr[\neg \tn{abort}_Q] \leq \frac{p}{\delta} \leq \frac{2^{-r}}{\delta},
\end{align}
and therefore $\Pr[\tn{abort}_Q] \geq 1 - \frac{2^{-r}}{\delta}$.

\vspace{1em}

In analogy to the above reasoning, one can construct an argument for a check in the $X$ basis. 
Therefore, we can write
\begin{align}\label{app:eq:Pr_W}
\begin{cases}
\tn{Pr} \left[w_{0,0} \in W_t^C| \neg \tn{abort}_Q \right] > 1-\delta' \\
\text{or}\\
\tn{Pr} \left[w_{0,0} \in W_t^C| \neg \tn{abort}_Q \right] \leq 1-\delta' \\
\qquad \text{ and } \Pr[\tn{abort}_Q] \geq 1 - \frac{2^{-r}}{\delta'}
\end{cases}
\end{align}
where $\delta'$ is a threshold for probability of measuring a string from $W_t^C$.

\vspace{1em}

Furthermore, in the protocol we verify that each of the $\ket{\bar{0}}$ ancilla states is sufficiently close to space $V_t$ before running the verification in the $X$ basis. Let $V_t^{0C}$ be a subspace of the code $V_t^C$ whose codewords are entries in the logical $\ket{\bar{0}}$, i.e. $0 + (W^{C*})_t$, where the dual code $(W^{C*})_t \subseteq V^C_t$.  Then $V_t^0$ is a subspace of $V_t$, such that $V_t^0$  is spanned by $\{\ket{v}: v \in V_t^{0C} \}$. Formally, we verify that conditioned on not aborting, the actual state of the ancilla is close to a codeword in $V_t^0$, or the verification phase aborts with high probability,
\begin{align}\label{app:eq:Pr_V0}
\begin{cases}
\tn{Pr} \left[v \in V_t^{0C}| \neg \tn{abort}_Q \right] > 1-\delta'' \\
\text{or}\\
\tn{Pr} \left[v \in V_t^{0C}| \neg \tn{abort}_Q \right] \leq 1-\delta'' \\
\qquad \text{ and } \Pr[\tn{abort}_Q] \geq 1 - \frac{2^{-r}}{\delta''}
\end{cases}
\end{align}
where $\delta''$ is a threshold for probability of measuring a string from $V_t^{0C}$. Since there are $r$ of ancilla checks, the probability that measuring all of the $\ket{\bar{0}}$ states yield a codeword from space $V_t^{0C}$ can be written as 
\begin{align}\label{qpp:eq:prob0}
\Pr[\bigwedge_{\ell=1}^r v_{\ell,0} \in V_t^{0C}\Big|\neg \tn{abort}_Q] \geq 1- r\delta''.
\end{align}
The purpose of having $\ket{\bar{0}} \in V_t^0$ is that using these ancillas for verification in the $X$ basis will not introduce bit flip errors in the $Z$ basis. In other words, any state in $V_t$ remains in $V_t$ after its verification in the $X$ basis, as long as we use ancillas $\ket{\bar{0}} \in V_t^0$.
 
\vspace{1em}

We will now make a statement about the whole verification phase. Let the state of the nodes after the verification in the $Z$ basis have the form
\begin{align} \label{app:eq:rho_ver}
\rho_{ver(Z)|b_Z \neq 0} = \alpha \rho_{V_t} + \beta \rho_{V_t^\perp} 
\end{align}
where $\rho_{V_t}$ is a mixture of pure states in $V_t$ and $\rho_{V_t^\perp}$ is a mixture of pure states in $V_t^\perp$. Here we condition the state on the fact that the public random bits $b_Z$ used in the verification in the standard basis (i.e. $b_{0,m}$ for $m = 1,\dots,r$) are all different than 0, i.e. that at least one CNOT gate is performed. In this case, measuring the state of the nodes after the CNOT, projects it either on $V_t$ or $V_t^{\perp}$. It happens with probabilities $\alpha$ and $\beta$, respectively.

Similarly, after the consecutive verification in the $X$ basis, the state of the nodes will be
\begin{align}\label{app:eq:rho_verZX}
\begin{split}
& \rho_{ver(Z,X)|b_Z, b_X \neq 0, \ket{\bar{0}}\in V_t^0} = \\
&~ = \alpha\alpha' \rho_{V_t \cap \mathcal{F}W_t} + \alpha\beta' \rho_{V_t^\perp \cap \mathcal{F}W_t} \\
&\qquad + \beta \left( \alpha'' \rho_{V_t \cap \mathcal{F}W_t^\perp} + \beta'' \rho_{V_t^\perp \cap \mathcal{F}W_t^\perp} \right) ,
\end{split}
\end{align}
where we additionally condition the state on the fact that bits $b_X$ used for verification in the $X$ basis are all different than zero (i.e. at least one CNOT was performed in the $X$ basis). Moreover, we condition it on the fact that $\ket{\bar{0}}$ ancillas used for verification in the $X$ basis are in $V_t^0$. Assuming the first lines of Equations \eqref{app:eq:Pr_V} and \eqref{app:eq:Pr_W}, we get that
\begin{align}
\alpha \alpha' + \alpha \beta' > 1 - \delta \\
\alpha \alpha' + \beta\alpha'' > 1 - \delta'
\end{align}
The first line implies that $\beta(\alpha'' + \beta'') \leq \delta$ and therefore, $\beta \leq \delta$. Using this in the second line we get that $\alpha \alpha' \geq 1 - \delta - \delta'$. Now, $\alpha \alpha'$ is exactly the probability that measuring $\rho_{ver(Z,X)|b_Z, b_X \neq 0, \ket{\bar{0}}\in V_t^0}$ in the $Z$ basis yields a string in $V_t^C$ \emph{and} measuring it in the $X$ basis yields a string in $W_t^C$.  Therefore, we get,
\begin{align}
\begin{split}
\Pr[v_{0,0}\in V_t^C \wedge w_{0,0} \in W_t^C | \neg \tn{abort}, b_Z, b_X \neq 0, \ket{\bar{0}}\in V_t^0   ] \\ \geq 1 - \delta - \delta'.
\end{split}
\end{align}

Now we will lower-bound the probability $\Pr[v_{0,0}\in V_t^C \wedge w_{0,0} \in W_t^C | \neg \tn{abort}  ] $ i.e.~remove the conditioning on $b_Z, b_X \neq 0, \ket{\bar{0}}\in V_t^0 $ from the above probability expression. Let us evaluate, 
\begin{widetext}
\begin{align}
\begin{split}
&\Pr[v_{0,0}\in V_t^C \wedge w_{0,0} \in W_t^C | \neg \tn{abort}] = \\
&= \Pr [b_Z, b_X \neq 0 \wedge \ket{\bar{0}}\in V_t^0 | \neg \tn{abort}] \Pr[v_{0,0}\in V_t^C \wedge w_{0,0} \in W_t^C | \neg \tn{abort}, b_Z, b_X \neq 0, \ket{\bar{0}}\in V_t^0   ] + \\
& + \underbrace{\Pr [\neg (b_Z, b_X \neq 0) \vee \ket{\bar{0}}\notin V_t^0 | \neg \tn{abort}]}_{\leq r2^{-r} + \Pr[\ket{\bar{0}}\notin V_t^0 | \neg \tn{abort} ] \leq r2^{-r} + r\delta''}
 \underbrace{\Pr[v_{0,0}\in V_t^C \wedge w_{0,0} \in W_t^C | \neg \tn{abort}, \neg (b_Z, b_X \neq 0), \ket{\bar{0}}\notin V_t^0   ] }_{\leq 1},
\end{split}
\end{align}
\end{widetext}
where we assumed the first line of Equation \eqref{app:eq:Pr_V0} to bound $\Pr[\ket{\bar{0}}\notin V_t^0 | \neg \tn{abort} ]$. To sum up, the conjunction of
\begin{align} \label{app:eq:cojunction}
\begin{split}
&\Pr \left[v_{0,0} \in V_t^C| \neg \tn{abort}_Q \right] > 1-\delta \\
&\Pr \left[w_{0,0} \in W_t^C| \neg \tn{abort}_Q \right] > 1-\delta' \\ 
&\Pr[\bigwedge_{\ell=1}^r v_{\ell,0} \in V_t^{0C}\Big|\neg \tn{abort}_Q] \geq 1- r\delta''
\end{split}
\end{align}
implies that
\begin{align} \label{app:eq:Pr_final}
\begin{split}
&\Pr[v_{0,0}\in V_t^C \wedge w_{0,0} \in W_t^C | \neg \tn{abort}] \geq \\
&\quad \geq (1-\delta - \delta') + r(2^{-r}+\delta'')(\delta + \delta').
\end{split}
\end{align}
Therefore, either Equation \eqref{app:eq:Pr_final}  is satisfied or at least one of the equations in \eqref{app:eq:cojunction} not satisfied. In the latter case, Equations \eqref{app:eq:Pr_V}, \eqref{app:eq:Pr_W} and \eqref{app:eq:Pr_V0} imply that 
\begin{align}\label{app:eq:Pabort}
\Pr[\tn{abort}] \geq 1 - \max \left\{ \frac{2^{-r}}{\delta}, \frac{2^{-r}}{\delta'}, \frac{2^{-r}}{\delta''} \right\}
\end{align}
\vspace{-1em}

\end{proof}

\begin{proof}[Proof of Lemma~\ref{lem:compl_vrqss}]
If the dealer is honest, the size of set $B$ must be at most $t$ -- there is at most $t$ malicious nodes and only real malicious nodes are accused of cheating. Therefore, the verification phase will always lead to accepting an honest dealer.

If $R$ is also honest then we must calculate the probability that the verification phase fails to identify the set $B$ of apparent malicious nodes. In this case, the reconstruction phase could take inconsistent shares to reconstruct the original state of the dealer. We can use the ``quantum-to-classical'' reduction argument again (see \cite{Smith2001} and the argument above) and argue about the probability of error for the classical protocol. An error in the classical case can occur when any of the checks for $Z$ or $X$ basis, or checks of $\ket{\bar{0}}$, lead to consistent strings on $V_t^C$, $\mathcal{F}W_t^C$ or $V_t^{C0}$.
Similarly to the argument above, the probability of that occurring is 
\begin{equation}\label{app:eq:completeness}
\epsilon_c = (2+r)2^{-r}
\end{equation}

Let us now look at the reconstruction phase of the quantum protocol to bound the fidelity of the output state. When the reconstructor is honest, she first applies a decoding operator to each branch $i$ corresponding to node $i \notin B$. The operator corrects errors without knowledge of the positions which carry errors (i.e. it corrects arbitrary errors). Therefore, whenever in qubits corresponding to branch $i\notin B$ there is no more than $t$ errors, the decoding will identify the errors and correct them. In the case when there are more than $t$ errors in a branch $i$, the procedure will leave that branch untouched and the reconstructor will update the set $B$ with position $i$.
Secondly, the honest reconstructor applies an erasure-recovery circuit to randomly chosen $n-2t$ positions from $i \notin B$. In the case when all of the errors are correctly identified in $B$, the erasure-recovery corrects for $n-2t$ erasure errors, i.e. missing qubits of the dealer and malicious nodes, and outputs the original state of the dealer. 
Since the verification phase can fail to identify the set $B$ with probability $\epsilon_c$, we have:
\begin{align}\label{eq:reconstruction}
\begin{split}
\rho_{rec}
 =~ & (1-\epsilon_c) \ketbra{\psi}+ \epsilon_c \tilde{\rho}_R,
\end{split}
\end{align}
where $\tilde{\rho}_R$ is an arbitrary state that depends on the action of the malicious nodes.
Let us define the fidelity of the reconstructed state as $F = \Tr \left[ \rho_{rec} \ketbra{\psi}_R \right]$. Using linearity properties of the trace together with the fact that quantum states have non-zero trace, we have that 
\begin{align}
\begin{split}\label{app:eq:fidelity}
F = & \Tr \left[ ((1-\epsilon_c) \ketbra{\psi} +  \epsilon_c \tilde{\rho})~ \ketbra{\psi} \right] \\
= & (1-\epsilon_c)\Tr \left[  \ketbra{\psi} \ketbra{\psi} \right] + \underbrace{\epsilon_c \Tr\left[ \tilde{\rho}~ \ketbra{\psi} \right]}_{\geq 0} \\
\geq & 1 - \epsilon_c.
\end{split}
\end{align}
 
\end{proof}

\end{document}